\documentclass[journal,12pt, draftclsnofoot, onecolumn,romanappendices]{IEEEtran}

\usepackage{bbm}
\usepackage{cite}
\usepackage{graphicx}  
\usepackage{caption}
\usepackage{subcaption}
\DeclareCaptionFont{tiny}{\tiny}
\usepackage{epstopdf}
\usepackage{amssymb,times}
\usepackage{centernot}

\usepackage{algorithm}
\usepackage{algorithmic}

\usepackage{multirow}
\usepackage{psfrag}

\usepackage[cmex10]{amsmath}
\usepackage{amssymb}
\usepackage{color}
\usepackage{amsthm}

\def\tr{\mathrm{Tr}}

\def\diag{\mathrm{diag}}

\newtheorem{theorem}{Theorem}

\newtheorem{proposition}{Proposition}

\usepackage[dvipsnames,svgnames]{xcolor}
\usepackage[textwidth=30mm]{todonotes}

\renewcommand{\baselinestretch}{1.45}
\usepackage{xparse}
\DeclareDocumentCommand\cc{ g g }{%
	\mathbb{C} \IfNoValueF {#1}    {^{#1%
			\IfNoValueF {#2} { \times #2}%
	}}%
}

\DeclareDocumentCommand\rr{ g g }{%
	\mathbb{R} \IfNoValueF {#1}    {^{#1%
		\IfNoValueF {#2} { \times #2}%
	}}%
}

\DeclareDocumentCommand\set{ m g g}{%
\{#1\}_{1\leq #2  	\IfNoValueF {#3} { \leq #3}}
}

\newcommand{\sbf}[1]{\boldsymbol{#1}}
\renewcommand{\bf}[1]{\mathbf {#1}}

\newcommand{\be}{\begin{equation}}
\newcommand{\ee}{\end{equation}}
\newcommand{\ba}{\begin{array}}
\newcommand{\ea}{\end{array}}
\newcommand{\bea}{\begin{eqnarray}}
\newcommand{\eea}{\end{eqnarray}}
\newcommand{\herm}{^{\mbox{\scriptsize H}}}

\newcommand{\tran}{^{\mbox{\scriptsize T}}}

\newcommand{\vbar}{\raisebox{.17ex}{\rule{.04em}{1.35ex}}}
\newcommand{\vbarind}{\raisebox{.01ex}{\rule{.04em}{1.1ex}}}
\newcommand{\R}{\ifmmode {\rm I}\hspace{-.2em}{\rm R} \else ${\rm I}\hspace{-.2em}{\rm R}$ \fi}
\newcommand{\T}{\ifmmode {\rm I}\hspace{-.2em}{\rm T} \else ${\rm I}\hspace{-.2em}{\rm T}$ \fi}
\newcommand{\N}{\ifmmode {\rm I}\hspace{-.2em}{\rm N} \else \mbox{${\rm I}\hspace{-.2em}{\rm N}$} \fi}
\newcommand{\B}{\ifmmode {\rm I}\hspace{-.2em}{\rm B} \else \mbox{${\rm I}\hspace{-.2em}{\rm B}$} \fi}
\newcommand{\Hil}{\ifmmode {\rm I}\hspace{-.2em}{\rm H} \else \mbox{${\rm I}\hspace{-.2em}{\rm H}$} \fi}
\newcommand{\C}{\ifmmode \hspace{.2em}\vbar\hspace{-.31em}{\rm C} \else \mbox{$\hspace{.2em}\vbar\hspace{-.31em}{\rm C}$} \fi}
\newcommand{\Cind}{\ifmmode \hspace{.2em}\vbarind\hspace{-.25em}{\rm C} \else \mbox{$\hspace{.2em}\vbarind\hspace{-.25em}{\rm C}$} \fi}
\newcommand{\Q}{\ifmmode \hspace{.2em}\vbar\hspace{-.31em}{\rm Q} \else \mbox{$\hspace{.2em}\vbar\hspace{-.31em}{\rm Q}$} \fi}
\newcommand{\Z}{\ifmmode {\rm Z}\hspace{-.28em}{\rm Z} \else ${\rm Z}\hspace{-.28em}{\rm Z}$ \fi}

\newcommand{\var}{\mbox {var}}

\renewcommand{\vec}[1]{\bf{#1}}     

\begin{document}


\title{
Capacity Approaching Low Density Spreading in Uplink NOMA via Asymptotic Analysis
	}

\author{Hossein Asgharimoghaddam \IEEEmembership{Student Member,~IEEE}, Jarkko Kaleva, \IEEEmembership{Member,~IEEE}, and  Antti T\"olli, \IEEEmembership{Senior~Member,~IEEE}
\thanks{ \newline \indent H. Asgharimoghaddam and A. T\"olli are with the Centre for Wireless Communications, University of Oulu, Oulu, Finland (Firstname.Lastname@oulu.fi). J. Kaleva (jarkko.kaleva@solmutech.com) is with the Solmu Technologies Oy. 
\newline\indent This work has been supported in part by the Academy of Finland 6Genesis Flagship (grant no. 318927).  
\newline \indent A preliminary version of this paper has been presented at the IEEE International Symposium on Information Theory (ISIT) 2020 Los Angeles, California, USA.
}
}

\maketitle
\vspace{-2cm}
\begin{abstract}
Low-density spreading non-orthogonal multiple-access (LDS-NOMA) is considered where $K$ single-antenna user-equipments (UEs) communicate with a base-station (BS) over $F$  fading sub-carriers. Each UE $k$ spreads its data symbols over $d_k<F$  sub-carriers. We aim to identify the LDS-code allocations that maximize the ergodic mutual information (EMI). The BS assigns resources solely based on pathlosses.  Conducting analysis in the regime where $F$, $K$, and ${d_k,\forall k}$ converge to $+\infty$ at the same rate, we present EMI as a deterministic equivalent plus a residual term. The deterministic equivalent is a function of pathlosses and spreading codes, and the small residual term scales as $\mathcal{O}(\frac{1}{\min(d_k^2)})$. We formulate an optimization problem to get the set  of all spreading codes, irrespective of sparsity constraints, which maximize the deterministic EMI. This yields a simple resource allocation rule that facilitates the construction of desired LDS-codes via an efficient partitioning algorithm. The acquired LDS-codes additionally harness the small incremental gain inherent in the residual term, and thus, attain near-optimal values of EMI in the finite regime. While regular LDS-NOMA is  found to be asymptotically optimal in symmetric models, an  irregular spreading arises in generic asymmetric cases. The spectral efficiency enhancement relative to regular and random spreading is validated numerically.
\end{abstract}


\IEEEpeerreviewmaketitle

\section{Introduction}
The objective in 5G and beyond network is to move on from merely connecting people to fully realizing the internet of things (IoT) and the fourth industrial revolution~\cite{5GEvol19NOKIA}. This has resulted in a shift towards techniques that both increase the spectral efficiency and support massive connectivity. Non-orthogonal multiple access (NOMA) as such a technique allows multiple user equipments (UEs) to share the same signal dimension via power domain (PD) or code domain (CD) multiplexing. PD-NOMA distinguishes UE's signals via superposition decoding principle and by exploiting the signal to interference plus noise ratio (SINR) difference among UEs. In CD-NOMA, including  low-density spreading (LDS) multiple access~\cite{hoshyar2008novel}, sparse code multiple access (SCMA)~\cite{SCMA14}, multi-user shared access (MUSA)~\cite{MUSIot16}, distinctive codes are assigned to UEs as in code division multiple access (CDMA) system~\cite{NOMAsurveyDing17,NOMAreviewDaiWang2015}.
The focus of this work is on LDS-NOMA, which employs an LDS code comprising a small number of $d_k$ non-zero elements for spreading UE $k$'s symbol over a number of $F$ shared radio resources~\cite{hoshyar2008novel,MCNomaUp14,LDSOFDM-Hoshyar10,LDSOFDMRazavi12}.
 Link level aspects of LDS-NOMA, e.g. bit-error-rate performance~\cite{hoshyar2008novel,LDSOFDM-Hoshyar10,LDSOFDMRazavi12,SCMA14} and envelop fluctuation~\cite{PeakToAvgImran10}, have been well-studied in the literature. However, the theoretical results about the boundaries of the achievable rates in LDS-NOMA is rather limited~\cite{ITLDSsurveyFerrant18}.
Here, a multi-carrier LDS-NOMA scheme with spreading in frequency domain~\cite{LDSOFDM-Hoshyar10,LDSReviewImran12} is considered. 
The objective is to characterize the boundaries of the ergodic sum rates at which UEs can jointly transmit reliably, and 
identify the LDS-code allocation policies that closely attain these boundaries.
Such analysis provides insight for system design~\cite{LDSReviewImran12,NOMAsurveyDing17}, and hence, is instrumental for the optimal utilization of scarce radio~resources.

\subsection{Prior related works}
The spectral efficiency (SE) analysis of LDS-NOMA with spreading in time domain has been considered in~\cite{shitz2017,zaidelBenj18,gerrante15SE,4036396,1633802,ITLDSsurveyFerrant18} under a symmetric AWGN channel model. These works compare the SE limits in the structured regular LDS codes and in the randomly generated irregular ones. 
The sparse mapping between numbers of $K$ UEs and $F$ resources in LDS-NOMA is called regular when each UE occupies a number of $d_k=d,\forall k$ resources and each resource is used by a number of $\frac{K}{F}d$ UEs; or irregular otherwise~\cite{shitz2017}. 
The irregular schemes with $d_k$ being randomly Poissonian distributed with fixed mean~\cite{4036396,1633802}, and randomly uniformly distributed~\cite{gerrante15SE} are studied using replica method~\cite{7079688} and random matrix framework developed in~\cite{verdu1999spectral}, respectively.
The regular scheme is considered in~\cite{shitz2017,zaidelBenj18}, where in~\cite{zaidelBenj18} a closed-form approximation is given for the SE limit. These analyses indicate that the regular codes, in symmetric AWGN channel, yield superior SE as compared to the irregular and the dense spreading (the case with $d_k=F, \forall k$) schemes.  
The aforementioned works rely on the analysis of random matrices in the large system regime~\cite{7079688,RMT}, where $F$ grows large with a fixed ratio of $K/F$. Such analysis yields rather accurate approximations in the finite regime that become arbitrarily tight as $F$ grows large.
However, since  the mathematical literature studying the limiting behavior of sparse random matrices is distinctly smaller than that for non-sparse ones~\cite{wood2012}, the extension of such SE analysis to more generic settings is rather difficult. The only large system analysis on the SE limit in presence of fading is considered in~\cite{FerranteFadingLDS18} 
in a special setting with $d_k=1, \forall k$. 
An information
theoretic analysis of LDS-NOMA with fading also appears in~\cite{RazaviHoshyarInfoLDS11}, wherein the ergodic sum-rate of random spreading has been evaluated numerically.
The application of large system analysis for studying the SE limits of spread spectrum system has been motivated in the pioneer works such as~\cite{MullerTulino04a,verdu1999spectral,TseZeitLinearCdma99,TseHanlyLinearCDMA99,GrantAlexOptimumCDMArandom98,GrantAlexRandoopticDMA96,shamaiverduCDMARand01}.
These works characterize the spectral efficiency of random dense-spreading ($d_k=F,\forall k$) CDMA with linear~\cite{TseZeitLinearCdma99,TseHanlyLinearCDMA99} and/or optimal~\cite{verdu1999spectral,GrantAlexOptimumCDMArandom98,GrantAlexRandoopticDMA96,shamaiverduCDMARand01, MullerTulino04a} receivers.
A common conclusion in a number of the aforementioned works is that  the random dense-spreading CDMA incurs negligible spectral efficiency loss relative to the optimum if an optimal receiver is used and the number of UEs per chip is sufficiently large~\cite{verdu1999spectral,GrantAlexOptimumCDMArandom98}.

\subsection{Contributions}
In this paper, a multi-carrier LDS-NOMA scheme with spreading in frequency domain~\cite{LDSOFDM-Hoshyar10,LDSReviewImran12} is considered. Different from the aforementioned works, the UEs are allowed to have distinct pathloss values, and frequency and time selective fading is assumed on the sub-carriers.
Also, instead of assuming a particular sparse mapping, we consider $d_k,\forall k$ as design parameters, and identify the LDS-code allocation policies that closely attain the maximum of
the ergodic mutual information (EMI). A key feature of these policies is that they assign the codes only based  on the pathloss values.
Conducting analysis in the large system limits where $F$, $K$, and ${d_k,\forall k}$ converge to $+\infty$ at the same rate,
we present EMI as a deterministic equivalent plus a residual term. The deterministic equivalent is given as a function of pathloss values and LDS-codes, and the small residual term is shown to quickly vanish inversely proportional to $d^2$ where $d=\min\{d_k,\forall k\}$.
First, we formulate an optimization problem to get the set  of all spreading codes, irrespective of the sparsity constraints, which maximize the deterministic equivalent of EMI. This yields a simple resource allocation rule that facilitates the construction of the desired sparse spreading codes via an efficient partitioning algorithm.
The analysis in the finite regime shows that the acquired sparse solutions additionally harness the small incremental gain inherent in the residual term, and thus, attain near-optimal values of the EMI in the finite regime.
It is observed that the regular spreading matrices are asymptotically optimal for the symmetric scenarios with the same pathlosses and power constraints for all UEs. However, in the generic asymmetric scenarios, an irregular structure might arise.
Numerical  simulations  validate  the  attainable spectral efficiency enhancement as compared to the random and the regular spreading schemes.

Parts of this paper have been published in the conference
publication~\cite{HosISIT020}. Specifically, the large system analysis and the optimization approach are sketched therein,
while the precise proofs along with the derivation details
are presented in the current work. In addition, the numerical analysis is extended to visualize the resource allocation policy proposed herein.

The remainder of this work is organized as follows. In
Section~\ref{sec: Problem statement}, the network
model and the problem formulation are given.
The proposed optimization approach based on large system analysis is presented in Section~\ref{sec: large system analysis}. The maximization of the deterministic EMI is considered in  Section~\ref{sec: opt code asymp}, and the gap to the optimum is characterized in Section~\ref{sec:asymp in finit}.
The algorithmic solution for the LDS codes assignment is presented in Section~\ref{sec:Alg max}. The numerical results are given in Section~\ref{sec:num res}. Conclusions are drawn in Section~\ref{sec:Conclusion} while all the proofs are presented in the Appendices.

\section{Problem Statement}
\label{sec: Problem statement}
\subsection{General Notations}
The following notations are used throughout the manuscript. All boldface letters indicate
vectors (lower case) or matrices (upper case). Superscripts
$(\cdot)\tran$, $(\cdot)\herm$,$(\cdot)^*$ ,$(\cdot)^{-1}$, $(\cdot)^{1/2}$
stand for transpose, Hermitian transpose, conjugate operator, matrix inversion and
positive semidefinite square root, respectively. We use $\mathbb{C}^{m
	\times n}$ and $\mathbb{R}^{m
	\times n}$ to denote the set of $m \times n$ complex and the real valued matrices, respectively.
 $\diag( \cdots)$ denotes a diagonal matrix with
elements $(\cdots)$ on the main diagonal. The $(i,j)^{\text{th}}$ element of a matrix $\bf A$ is denoted by $[{\bf A}]_{i,j}$ or $a_{i,j}$. 
The sets are indicated by
calligraphic letters. The  cardinality
of a set $\mathcal{A}$ is denoted by $|\mathcal{A}|$, and $\mathcal{A}\backslash k
$ is used to exclude the index $k$ from the set.
${\mathbb{E}}\{\cdot\}$,
${\tr\{\cdot\}}$ denote  statistical
expectation and  trace operator,
respectively.  Euclidean (spectral) norm for vectors (matrices) are denoted by $\|\cdot\|$. The notation $|.|$ is used to denote both the absolute value for a complex scalar, and the determinant for a square matrix. Notation $x_F=\mathcal{O}(\alpha_F)$ represents inequality $|x_F|\leq C \alpha_F$ as $F\rightarrow \infty$ with $C$ being a generic constant independent of system size $F$.

\subsection{System Model}
\label{sec:sys model}
   \begin{figure}[tpb]
	\begin{center}
		\includegraphics[clip, trim=0cm 0.0cm 0cm 0cm, width=\columnwidth]{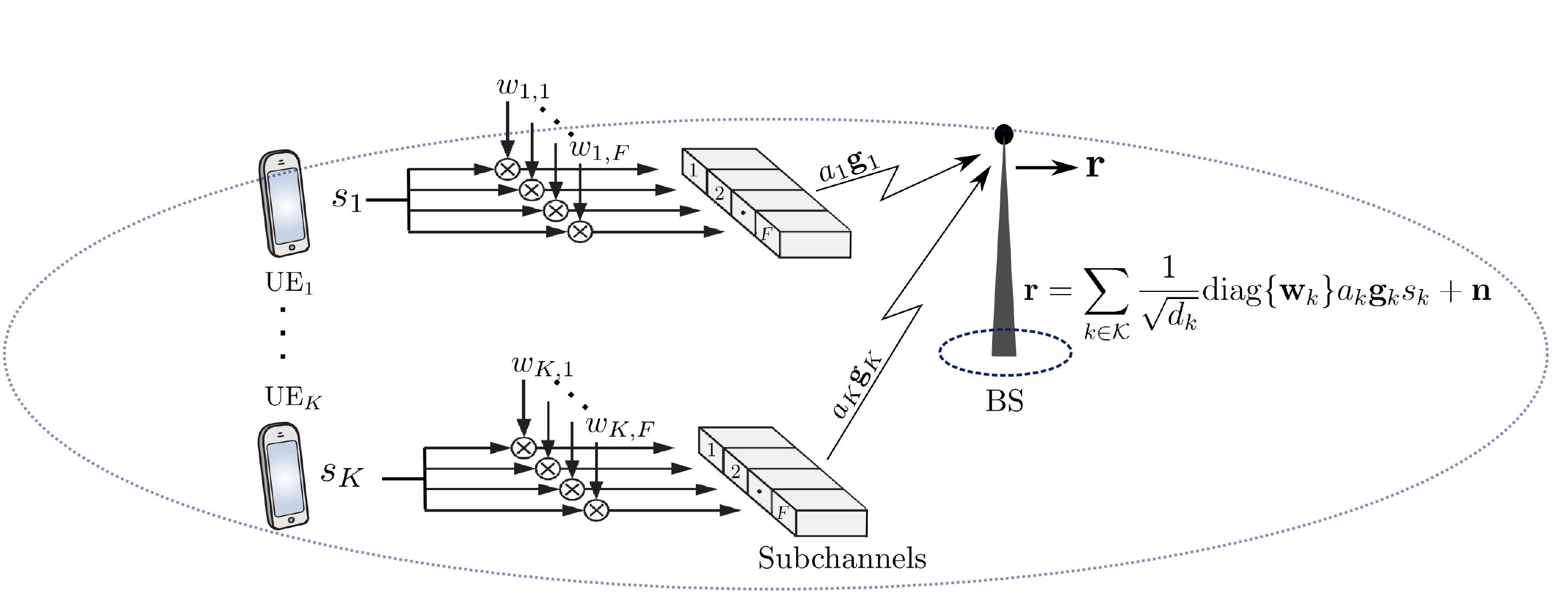}
	\end{center}\vspace{-0.5cm}
	\caption{The system model illustration.
	}
	\label{fig:sysmdl}
\end{figure}
Consider an uplink multi-carrier system with $K$ single-antenna UEs transmitting to a base station (BS) on a common frequency band. The set of UEs' indices is denoted by $\mathcal{K}$ hereafter.
Exploiting the  OFDM technique, the total frequency band is divided into a set of narrow band sub-channels $\mathcal{F} = \{1,. . .,F\}$. Then, each UE $k$ spreads its data symbol in frequency direction using a low density spreading code ${\vec{w}}_k \in \rr{F}$. 
The code is a sparse vector consisting of $F$ chips with $d_k$, a small number, of non-zero values.
The UEs' codes are not restricted to be orthogonal, and thus, the signals of UEs transmitting on the same sub-channel will be superimposed. 
  The received signal vector ${\bf r}\in \cc{F}$ is given as 
  \begin{equation}
 \label{eq:resive vec}
 \mathbf{r}=\sum_{k\in \mathcal{K}} \frac{1}{\sqrt{d_k}}\mathrm{diag}\{{\bf w}_{k}\} a_k{\bf g}_k  s_k+{\bf n}
 \end{equation}
  where the $f^{\text{th}}$ element of the vector~$\bf r$ corresponds to the signal received within $f^{\text{th}}$ sub-channel.  The noise vector is given by ${\bf n}\sim\mathcal{CN}(0,\sigma^2 {\bf I}_{F})$. The unit variance symbol of $k^{\text{th}}$ UE is denoted by~$s_k$.
 The channel vector for UE $k$ is denoted by $a_k {\bf g}_k\in \cc{F}$, defined in more detail in the sequel. 
   The  vector ${\bf w}_k\in\rr{F}$ denotes the UE $k$'s spreading code, and is assumed to satisfy the transmit
   power condition
   \begin{equation}
   \label{eq:pwr constraint}
   \frac{1}{d_k}\sum_{f\in \mathcal{F}} {|w_{f,k}|^2}\leq P_k , \forall \, k\in\mathcal{K}
   \end{equation}
 where $P_k$ is the total power available at UE $k$.  The normalization by $d_k$ values in~\eqref{eq:resive vec} and~\eqref{eq:pwr constraint} is considered to make the sparsity assumption in the system model explicit.

\subsection{Channel Model}
\label{sec:ch mdl}
The channel matrix entries are generated based on the uncorrelated fading channel model utilized in the context of multi-carrier systems~\cite[Chapter 1]{Kaiserbook}.
This channel model is based on the assumption that 
the fading on the adjacent data symbols
after inverse OFDM and de-interleaving can be considered as uncorrelated~\cite{Kaiserbook}. This
assumption holds when, for example, 
time and frequency interleavers with sufficient
interleaving depth are applied. 
 Thus, the resulting complex-valued channel fading coefficient
is  generated independently for each sub-carrier and OFDM symbol. For a
propagation scenario  without line of sight, the fading amplitude is generated
by a Rayleigh distribution and the channel model is referred to as uncorrelated
Rayleigh fading channel~\cite{Kaiserbook}. In particular,  we present the channel vector for UE $k$ as $a_{k} {\mathbf{g}}_{k}$ where  $a^2_{k}$ includes the pathloss due to large scale fading, and the matrix ${\bf G}=[{\bf g}_1,...{\bf g}_K]$ represents the small-scale fading. 
The entries of ${\bf G}$ are independent complex Gaussian random variables. Each entry has zero-mean independent real and imaginary parts with variance of $\frac{1}{2}$.
In the following, we use ${\bf h}_k=\frac{a_k}{\sqrt{d_k}} \mathrm{diag}\{{\bf w}_{k}\} {\bf g}_k$ and ${\bf H}=[{\bf h}_1,...{\bf h}_K]$ to denote the equivalent channel vectors and matrix, respectively, including the spreading vectors.

\subsection{Ergodic Capacity of the Channel} 
Let ${\bf W} \in \mathbb{R}^{F\times K}$ to be the spreading matrix that contains all the spreading vectors $\{{\bf w}_k\}$ in its columns. We define the set of feasible sparse spreading matrices as  $\mathcal{C}_1=\{{\bf W}| w_{f,k}\in \mathbb{R}, \frac{1}{d_k}\sum_f |w_{f,k}|^2\leq P_k, \|{\bf w}_k\|_{0}=d_k,\,\forall k\in\mathcal{K}\}$ where $\|{\bf w}_k\|_{0}=d_k$ restricts the number of non-zero elements in ${\bf w}_k$ to be equal to $d_k$.
Our  interest  is  in the  scenario  where   the  BS  assigns  the  spreading  codes  based  on  the  UEs'  pathlosses, and the assigned spreading codes are
known to the UEs.
Implicit in this
model is the assumption that the channel statistics vary much more
slowly than the small-scale fading coefficients, so that the statistical properties of the channel can be assumed constant for a long period of communication~\cite{BjornsonLuca2016}. 
Given a perfect knowledge of fading coefficients at the BS side, the ergodic mutual information (EMI) between the transmitters and the receiver, for a spreading matrix  ${\bf W}\in \mathcal{C}_1$, is
\begin{equation}
\label{eq:Mutual info}
\begin{aligned}
{\mathrm J}({\bf W},\sigma^2 )=\frac{1}{F}\mathbbm{E}\log\big|{\bf I}_{F}+\frac{1}{\sigma ^2}{{\sum_{k\in \mathcal{K}}\frac{a_k^2}{ d_k}  \mathrm{diag}\{{\bf w}_{k}\} {\bf g}_k {\bf g}_k\herm  \mathrm{diag}\{{\bf w}_{k}\} }}\big|\,\\
\end{aligned}
\end{equation}
where the expectation is over the small-scale fading only.
 We are particularly interested in the ergodic capacity of the
channel, which is equal to the maximum of ${\rm J}({\bf W},\sigma^2)$ over the
set of  all the sparse spreading matrices in $\mathcal{C}_1$, i.e.,
\begin{equation}
\label{eq:optimizeMutualInfo}
{\rm C}_{\rm E}(\sigma^2)= \underset{{\vec W}\in \mathcal{C}_1}{\max} {\rm J}({\bf W},\sigma^2).
\end{equation}
The corresponding conventional  problem without sparsity constraints  is considered in~\cite{Rupf94} for symmetric additive white Gaussian noise (AWGN) channel model. It is shown in~\cite{Rupf94} that Welch-Bound-Equality (WBE) signature sequences achieve the sum-capacity of symmetric AWGN channel. The optimal spreading codes in asymmetric AWGN model is studied in~\cite{Viswanath99} via the concept of Majorization and Schur-concavity~\cite{Marshallbook79} of the sum-capacity with respect to eigenvalues of ${\bf H}{\bf H}\herm$. While AWGN channel capacity
  depends on the spreading codes
through their cross-correlations~\cite{Rupf94,Viswanath99}, the  transmission of the signal over the i.i.d Rayleigh fading channel destroys the orthogonality of the spreading codes~\cite{ReinhardtCDMAFading95}.
This, together with the hypotheses on the availability of only channel statistics, allow  the spreading codes to  be taken from the set of vectors $\{{\bf w}_k\}$ with positive values and independent of the small-scale fading gains, as stated in the following proposition,
\begin{proposition}
	\label{prop:sign}
	The expectation in~\eqref{eq:Mutual info} is invariant with respect to the signs of the entries in real valued spreading vectors ${\bf w}_k,\,\forall k\in \mathcal{K}$.
\end{proposition}
\begin{proof} 
	The proof follows by showing that a column of ${\mathbf{H}}$, given as ${\bf h}_k=\frac{a_k}{\sqrt{d_k}} \mathrm{diag}\{{\bf w}_{k}\} {\bf g}_k$, has the same distribution as $\frac{a_k}{\sqrt{d_k}}\mathrm{diag}\{|{ w}_{1,k}|,...,|{ w}_{F,k}|\} {\bf g}_k$. The claim follows directly from the invariance of  i.i.d complex Gaussian vectors in distribution under unitary transformation~\cite{AWGNCApTeltar99}.
\end{proof}
Even though the cross-correlation properties of the spreading codes is not a determining parameter in~\eqref{eq:optimizeMutualInfo},  the optimal pairing of UEs and the power loading on each sub-channel need to be studied. 
The sparsity requirement for spreading codes  impose binary constraints in~\eqref{eq:optimizeMutualInfo} that makes the problem non-convex. Also, the convexity of the objective function in~\eqref{eq:optimizeMutualInfo} cannot be verified due to the expectation operator. 
However, the objective function without the expectation operator can be shown to be non-convex.
Moreover, the expectation in~\eqref{eq:optimizeMutualInfo} needs to be evaluated in a concise form and in terms of the spreading codes and the pathloss values.

\section{An optimization approach based on asymptotic analysis} 
\label{sec: large system analysis}
In the following, we use theory of large random matrices~\cite{RMT} to characterize the EMI in terms of the spreading codes and the pathloss values. The large system analysis of the problem is carried out in the asymptotic regime where 
$F\rightarrow \infty$ with $KF^{-1}\in (0,\infty)$ and $d_k F^{-1}\in (0,1],\forall k$.
 The limiting results yield rather accurate approximations for the finite-size scenarios~\cite{RMT}. 
In deriving the large system analysis, we use subscript $F$ to denote the dependency of the entities on the system size.

\begin{theorem}
	\label{th:cap conv}
	Consider the channel matrices ${\bf H}_{F}=[{\bf h}_1,...,{\bf h}_K]$ with ${\bf h}_{k}=\frac{a_k}{\sqrt{d_k}} \mathrm{diag}\{{\bf w}_{k}\} {\bf g}_{k}, \forall k \in\mathcal{K}$. 
	The entries of ${\bf G}_{F}=[{\bf{g}}_1,...,{\bf{g}}_{K}]$  are i.i.d standard complex Gaussian random variables.
	The deterministic vectors ${\bf w}_k,\forall k \in \mathcal{K}$ are the columns of ${\bf W}_F\in \mathcal{C}_1$, each with $d_k$ non-zero values.  The scalars $\set{a_k}{k}{K}$ are bounded real values.
	Then, as $F\rightarrow \infty$ with $KF^{-1}\in (0,\infty)$ and $d_k F^{-1}\in (0,1],\forall k$, the EMI converges to a deterministic equivalent such that
	\begin{equation}
	\label{eq:Converg CAp}
	{\mathrm J}_{F}({\bf W}_F,\sigma^2)=\bar{\rm J}_F({\bf W}_F,\sigma^2)+\epsilon_{F}
	\end{equation}
	where $\epsilon_F=\mathcal{O}(\frac{1}{d^2})$,
	$d=\min\set{d_k}{k}{K}$, and 
	\begin{equation}
	\label{eq:asymp I theorem}
	\begin{aligned}
	&\bar{{\rm J}}_F({\bf W}_F,\sigma^2)=\frac{1}{F}\sum_{k\in \mathcal{K}} \log(1+\frac{ a_k^2 }{\sigma^2{d_k}}\sum_{f\in\mathcal{F}} {w_{f,k}^2 r_{f}})\\
	& \!\!\!\!+\frac{1}{F}\sum_{f\in\mathcal{F}} \log(1+\frac{1}{\sigma^2}\sum_{k\in\mathcal{K} } \frac{1}{d_k}{ w_{f,k}^2a_k^2\tilde{r}_{k}})-\frac{1}{\sigma^2 F}\sum_{f\in\mathcal{F}}\sum_{k\in\mathcal{K}}\frac{w^2_{f,k}}{d_k}a_k^2 r_{f}\tilde{r}_k
   \end{aligned}
    \end{equation}
where $r_f({\bf W}_F,\sigma^2)$ and $\tilde{r}_k({\bf W}_F,\sigma^2)$ are the solutions of 
\begin{equation}
\label{eq:rf and rk}
\begin{aligned}
&r_{f}=(1+\frac{1}{\sigma^2}\sum_{k\in\mathcal{K}} \frac{1}{d_k}w^2_{f,k} a_k^2 \tilde{r}_k  )^{-1}, \forall f\in\mathcal{F},\\
&\tilde{r}_k=(1+  \frac{ a_k^2}{\sigma^2{d_k}}   \sum_{f\in\mathcal{F}} w^2_{f,k}  {r}_{f})^{-1},\forall k\in\mathcal{K}.
\end{aligned}
\end{equation}
\end{theorem} 
\begin{proof}
The proof of the theorem is given in Appendix~\ref{sec:prof of the 1} wherein we use an Integration by parts formula~\cite{Hachem07} to derive an expression for the expectation of the mutual information as in~\eqref{eq:Converg CAp}. Then,
we derive an upper-bound for $\epsilon_{F}$  using Nash-Poincar\'{e} inequality~\cite{Hachem07} where the convergence rate of $\mathcal{O}(\frac{1}{d^2})$ is claimed accordingly.
The convergence ${\mathrm J}_{F}({\bf W}_F,\sigma^2)-\bar{\rm J}_F({\bf W}_F,\sigma^2)\rightarrow 0$ can be also claimed  relying on Girko's law~\cite[Section 3.2.3]{Mller2013ApplicationsOL}~\cite[Theorem 6.10]{RMT}.
An alternative proof based on Replica method is also given in~\cite{ReplicamethodMullerBook}. 
Regarding the convergence rate, it is shown in~\cite{hachemCLT2008} that in the case where ${\bf h}_k,\forall k\in\mathcal{K}$ are Gaussian vectors with given variance profiles, the convergence rate is  $\mathcal{O}(\frac{1}{F^2})$. One might be able to obtain the convergence rate declared in the theorem by properly scaling the variances in~\cite{hachemCLT2008} while ensuring that the assumptions therein remain valid. However, for
convenience of the reader and to avoid ambiguity, a straightforward proof of the theorem is presented in Appendix~\ref{sec:prof of the 1} based on an
alternative technique, known as the Gaussian method~\cite{Hachem07,RMT}, which is particularly
suited to random matrix models with Gaussian~entries.
\end{proof}

According to Theorem~\ref{th:cap conv}, the EMI ${\mathrm J}_F({\bf W}_F,\sigma^2)$ converges asymptotically to the deterministic equivalent $\bar{\mathrm J}_F({\bf W}_F,\sigma^2)$ with a convergence rate of $\mathcal{O}(\frac{1}{d^2})$. In the finite scenarios of interest with a moderate number of sub-channels, $d_k$ values can be small relative to $F$.\footnote{Note that while the limiting results are obtained in the asymptotic regime, those can be applied as approximations for the finite scenarios with dimensions as small as 8 and even 4 or 2~\cite[Section 2.2.1]{RMT}.} In such finite cases, the analysis in  Section~\ref{sec:asymp in finit} shows that the residual term $\epsilon_F$ appears as a small incremental gain in the EMI of the sparse spreading scheme, which is dictated mainly by the number of non-zero elements in the codes.
Keeping this in mind, we propose an optimization approach as in the following. In Section~\ref{sec: opt code asymp}, we first formulate an optimization problem to get the set $\bar{\mathcal{C}^*}$ of all power constrained spreading codes, irrespective of the sparsity constraints, which maximize the deterministic EMI $\bar{\mathrm J}_F({\bf W}_F,\sigma^2)$. This yields a simple resource allocation rule that facilitates the construction of the desired sparse spreading codes via an efficient partitioning algorithm. The details about the partitioning algorithms is delegated to Section~\ref{sec:Alg max}. In Section~\ref{sec:asymp in finit}, we show that the sparse solutions in $\bar{\mathcal{C}^*}$ additionally harness the  small incremental gain inherent in the residual term $\epsilon_{F}$ in the finite regime.
The analysis in Section~\ref{sec:RMT ergotic EMI} eventually yields an upper-bound on the gap to the optimum, which is shown to be close to zero for the sparse solutions in $\bar{\mathcal{C}^*}$.


\section{Maximizing the ergodic mutual information}\label{sec:RMT ergotic EMI}
In the sequel, we omit the subscript $F$ denoting the dependency on system size. Also, observe that $\bar{{\rm J}}({\bf W},\sigma^2)$ in~\eqref{eq:asymp I theorem} depends only on squares of $w_{f,k}$ values. Therefore, with a change of variable $v_{f,k}=\frac{1}{d_k}w_{f,k}^2$, hereafter, we  express the EMI and the related entities as a function of matrix ${\vec V}=[v_{f,k}]_{f\in\mathcal{F},k\in\mathcal{K}}$. Given a matrix ${\vec V}$, the corresponding spreading vectors $\{{\bf w}_k\}$ can be obtained up to an uncertainty in the signs of the entries in the spreading vectors. It is shown in Proposition~\ref{prop:sign}  that the objective function under the considered i.i.d channel model is indifferent to the signs of the spreading code entries. Thus, hereafter, we refer to $\bf V$ and $\bf W$ interchangeably as the spreading~matrix.

\subsection{The optimal spreading in the asymptotic regime}
\label{sec: opt code asymp}
We first neglect the sparsity constraints, and define  $\mathcal{C}_2\triangleq\{{\bf V}| {v}_{f,k}\in \mathbb{R}^+, \sum_{f\in\mathcal{F}} v_{f,k}\leq P_k,\,\forall k\in\mathcal{K}\}$ to be the set of all the power constrained spreading matrices $\vec{V}$. Then, we formulate the problem of maximizing the deterministic EMI $\bar{\rm J}({\bf V},\sigma^2)$ as follows
\begin{subequations}
	\label{eq:Opt Asymp Rewriten}
\begin{align}
&\underset{{\vec V}\in\mathcal{C}_2}{\max} \quad \bar{{\rm J}}({\bf V},\sigma^2)\\
&\,\,\text{s.t.}\quad \sum_{f \in \mathcal{F}} v_{f,k}\leq P_k , \forall k\in \mathcal{K}\label{eq:Constrain 2},\\
&\quad \quad\,\,\,\, v_{f,k}\geq 0,\,\forall k\in \mathcal{K},f \in \mathcal{F}.\label{eq:Constrain 3}
\end{align}
\end{subequations}
This optimization problem yields the set of all spreading codes that maximize the deterministic EMI $ \bar{\rm J}({\bf V},\sigma^2)$ subject to the power constraints and irrespective of the sparsity constraints. 

The Karush-Kuhn-Tucker (KKT) conditions~\cite{Boyd-Vandenberghe-04} are necessary conditions for a matrix ${\bf V} $ to be a local optimal solution of the problem in~\eqref{eq:Opt Asymp Rewriten}. However, the KKT conditions are not necessarily the sufficient conditions. The sufficiency and the globally optimality of the solutions are discussed later.
The Lagrangian associated with~\eqref{eq:Opt Asymp Rewriten} is given as
\begin{equation}
L(v_{i,j},\lambda_{i,j},\delta_{j})=-\bar{\rm J}({\bf V},\sigma^2)-\sum_{i\in\mathcal{F}}\sum_{j\in \mathcal{K}}\lambda_{i,j}v_{i,j}+\sum_{j\in \mathcal{K}}\delta_j \big( \sum_{i \in \mathcal{F}} {v}_{i,j}-P_j\big)
\end{equation} 
where the Lagrangian variables  $\delta_j$ and $\lambda_{i,j}$ are associated with constraints~\eqref{eq:Constrain 2} and~\eqref{eq:Constrain 3}, respectively.
The gradient of the Lagrangian can then be evaluated as
\begin{equation}
  \nabla_{\!\!_{f,k}}L(v_{i,j},\lambda_{i,j},\delta_{j}) =-\frac{\partial \bar{{\rm J}}({\bf V},\sigma^2)}{v_{f,k}}-\lambda_{f,k}+\delta_{k}, \quad\forall f\in\mathcal{F},k\in \mathcal{K}.
\end{equation}
Note that $\bar{{\rm J}}({\bf V},\sigma^2)$ depends on the entries of $\bf V$ via $r_f({\bf V},\sigma^2)$ and $\tilde{r}_k({\bf V},\sigma^2)$ as in~\eqref{eq:rf and rk} and~\eqref{eq:asymp I theorem}. Since those are the solutions to the saddle point equations, the partial derivatives $\frac{\partial \bar{{\rm J}}({\bf V},\sigma^2)}{\partial r_{f}}$ and $\frac{\partial \bar{{\rm J}}({\bf V},\sigma^2)}{\partial \tilde{r}_{k}}$ are zero at any point given by $({\bf V}, r_f({\bf V},\sigma^2),\tilde{r}_k({\bf V},\sigma^2))$.\footnote{One can verity this by evaluating the partial derivatives using~\eqref{eq:rf and rk} and~\eqref{eq:asymp I theorem}.}
 Therefore, the chain rules of derivatives~\cite{AntonChainRule95} allow the partial derivative $\frac{\partial \bar{{\rm J}}({\bf V},\sigma^2)}{v_{f,k}}$ to be evaluated by assuming $r_f$ and $\tilde{r}_k$ as constants. This, in particular, gives $
\frac{\partial \bar{{\rm J}}({\bf V},\sigma^2)}{v_{f,k}}=\frac{1}{\sigma^2F}a_k^2 \tilde{r}_k r_f$, which yields the KKT conditions as
\begin{equation}
\begin{aligned}
& \lambda_{f,k}^*\geq 0,\quad \lambda_{f,k}^* v_{f,k}=0,\quad \delta_k^*\geq 0,\quad \delta_k^* \big( \sum_{i \in \mathcal{F}} {v}_{i,k}-P_k\big)=0,\quad\forall f\in\mathcal{F},k\in \mathcal{K},\\
&\quad\quad\quad\quad\quad-\frac{1}{\sigma^2F}a_k^2 \tilde{r}_k r_f -\lambda_{f,k}^*+\delta_k^*=0,\quad\forall f\in\mathcal{F},k\in \mathcal{K}
\end{aligned}
\end{equation}     
where $\lambda_{f,k}^*$ and $\delta_k^*$ denote the optimal values of the Lagrangian variables.
Since $\lambda_{f,k}$ can be solved from the last equation, the KKT conditions can be simplified as
\begin{subequations}
	\label{eq:KKT core}
\begin{align}
& \delta_k^* \big( \sum_{i \in \mathcal{F}} {v}_{i,k}-P_k\big)=0\label{eq:KKT1},\quad \forall k\in \mathcal{K},\\
&(\delta_k^*-\frac{1}{\sigma^2F}a_k^2 \tilde{r}_k r_f )v_{f,k}=0\label{eq:KKT2},\quad\forall f\in\mathcal{F},k\in \mathcal{K},\\
&\delta_k^*\geq \frac{1}{\sigma^2F}a_k^2 \tilde{r}_k r_f,\quad\forall f\in\mathcal{F},k\in \mathcal{K}. \label{eq:KKT3}
\end{align}
\end{subequations}  
According to the KKT conditions,
some properties for the optimal solutions to~\eqref{eq:Opt Asymp Rewriten} can be summarized as in the following proposition.
\begin{proposition}
	\label{prop:KKT properties}
The spreading matrices that maximize the ergodic mutual information in~\eqref{eq:Opt Asymp Rewriten} have the following properties:
\begin{itemize}
	\item The power constraints in~\eqref{eq:Constrain 2} are satisfied with equality for all UEs, i.e., all UEs are active, and transmit with full power.
	\item The parameters $r_f$ are equal to $r^*,\,\forall f\in \mathcal{F}$ where $r^*$ is the solution of the following fixed point iterations
	\begin{equation}
	\label{eq:opt r}
	r^*=\big(1+\frac{1}{F} \sum_{k\in\mathcal{K}} \frac{ P_k a_k^2}{\sigma^2+{{P_k}}a_k^2 r^*}   \big)^{-1}.
	\end{equation}
	\item The parameters $\tilde{r}_k,\forall k$ are equal to $\tilde{r}_k^*,\forall k$ where 
	\begin{equation}
	\tilde{r}_k^*=\frac{{\sigma^2}}{{\sigma^2}+  {{P_k}} a_k^2  {r}^*},\forall k\in \mathcal{K}.
	\end{equation}
\end{itemize}
\end{proposition}
\begin{proof}
 The proof is given in Appendix~\ref{sec: Proof of kkt propertis}.
\end{proof}
As a result of Proposition~\ref{prop:KKT properties}, the solutions satisfying the KKT conditions 
 must give $r_f=r^*,\forall f\in\mathcal{F}$ and $\tilde{r}_k=\tilde{r}_k^*,\forall k\in \mathcal{K}$.
  Observe that the values of $r^*$ and $\tilde{r}_k^*$ are given independently from the values of $v_{f,k}$. Let $\bar{{\rm J}}({\bf V},\tilde{r}_k^*,r^*,\sigma^2)$ denotes a function obtained by plugging the $r^*$ and $\tilde{r}_k^*$ values into~\eqref{eq:asymp I theorem}, i.e.,
  \begin{equation}
      \begin{aligned}
      \bar{{\rm J}}({\bf V},\tilde{r}_k^*,r^*,\sigma^2)&=-\frac{1}{F}\sum_{k\in \mathcal{K}} \log(\tilde{r}^*_{k})
-\log(r^* )  -\frac{r^*}{\sigma^2F}\sum_{k\in\mathcal{K}}a_k^2\tilde{r}^*_k \sum_{f\in\mathcal{F}}v_{f,k},\\
 &=-\frac{1}{F}\sum_{k\in \mathcal{K}} \log(\tilde{r}^*_{k})
-\log(r^* )  -\frac{r^*}{\sigma^2F}\sum_{k\in\mathcal{K}}{P_k}a_k^2\tilde{r}^*_k
      \end{aligned}
  \end{equation}
  where the arguments of the logarithms are replaced with their equivalents from~\eqref{eq:rf and rk}, and in the last equality we used the first property from Proposition~\ref{prop:KKT properties}. Observe that the values of $\bar{{\rm J}}({\bf V},\tilde{r}_k^*,r^*,\sigma^2)$ depend only on $\tilde{r}_k^*$ and $r^*$ values. Thus, all the solutions of~\eqref{eq:Opt Asymp Rewriten}, satisfying the KKT conditions, attain the same value of the objective function, i.e., its  global maximum value.  This proves the sufficiency of the KKT conditions for a spreading matrix $\bf V$ to be the optimal solution of~\eqref{eq:Opt Asymp Rewriten}. 
As a result of the KKT conditions, the spreading codes $\vec{v}_k$, which maximize $\bar{{\rm J}}({\bf V},\sigma^2)$ in~\eqref{eq:Opt Asymp Rewriten}, can be evaluated as the positive solutions of the following indeterminate system of equations
\begin{subequations}
\label{eq:system of underdetermined eq}
\begin{align}
&\sum_{f \in \mathcal{F}} v_{f,k}=P_k,\,\forall k\in\mathcal{K},\label{eq:system of underdetermined eq a}\\
&\sum_{k\in\mathcal{K}} \beta_k v_{f,k}=\frac{1}{r^*}-1,\,\forall f\in\mathcal{F}\label{eq:system of underdetermined eq b}
\end{align}
\end{subequations}
where $\beta_k=\frac{ a_k^2}{\sigma^2+{{P_k}}a_k^2 r^*}$.
In these equations, $r^*$ is a fixed scalar, which is evaluated from~\eqref{eq:opt r}.
The equalities in~\eqref{eq:system of underdetermined eq b} are obtained by setting $r_f=r^*,\forall f$ in~\eqref{eq:rf and rk}. 
These equalities follow since any spreading matrix  that gives $r_f=r^*,\forall f$, equivalently, satisfies the second and third KKT conditions in~\eqref{eq:KKT core} as well. The equalities in~\eqref{eq:system of underdetermined eq a} are given as a result of the first KKT condition~in~\eqref{eq:KKT core}.

The system of equations in~\eqref{eq:system of underdetermined eq} has a simple implication. The first line of equalities in~\eqref{eq:system of underdetermined eq a} indicates that the UEs need to transmit with full power. The second line implies that the entries $v_{f,k}$ should be assigned such that $r_f$ values become the same across all sub-channels. One can verify that the dense spreading, i.e., $v_{f,k}=\frac{P_k}{F},\forall f,k$ is a solution of~\eqref{eq:system of underdetermined eq}.  This has been also observed in~\cite{MullerTulino04a} where the authors show that in a scenario with randomly assigned dense spreading sequences, the frequency-dependency of $r_f$ values vanishes asymptotically. Aside from the dense spreading matrix, sparse spreading matrices can be also designed to satisfy \eqref{eq:system of underdetermined eq}. To this end, the non-zero elements $v_{f,k}$ should be assigned to the sub-channels such that the weighted sum $\sum_{k\in\mathcal{K}} \beta_k v_{f,k}$ becomes the same, i.e., equal to $\frac{1}{r^*}-1$, across all sub-channels.
 In the symmetric case with $a_{k}=a$, $P_k=P$ and $d_k=d,\forall k$, the second line of equalites in~\eqref{eq:system of underdetermined eq b} becomes $\frac{a^2}{\sigma^2+Pa^2 r^*} \sum_{k\in \mathcal{K}}  v_{f,k}=\frac{1}{r^*}-1,\forall f\in\mathcal{F}$. Given $dK/F$ to be integer, one can verity that any regular spreading matrix with non-zero elements being $\frac{P}{d}$ is a solution to~\eqref{eq:system of underdetermined eq}. However, in the generic non-symmetric case, an irregular assignment of the non-zero values might arise to ensure the conditions in~\eqref{eq:system of underdetermined eq} to hold. Such an assignment for the generic case is done in Section~\ref{sec:Alg max} via a simple partitioning algorithm.

 \subsection{On the optimality of the asymptotic sparse spreading codes}
\label{sec:asymp in finit}
Based on the analysis in Section~\ref{sec: opt code asymp}, we know that 
the  power-constrained spreading codes in $\mathcal{C}_2$ that maximize the deterministic equivalent $\bar{\rm J}({\bf V},\sigma^2)$ in~\eqref{eq:Opt Asymp Rewriten} are given as the solutions of the system of equations in~\eqref{eq:system of underdetermined eq}.
This set of solutions has been defined as $\bar{\mathcal{C}}^*$. Hereafter, we use $\bar {\bf V}^*$ to refer to  a member of the set $\bar{\mathcal{C}}^*$.
The analysis in Theorem~\ref{th:cap conv} shows that the residual term $\epsilon$  
vanishes with a rate inversely proportional to the square of the number of non-zero elements in the codes. Thus, the solutions in $\bar{\mathcal{C}}^*$ attain the maximum of EMI ${\mathrm J}({{\bf V}},\sigma^2)$ in the asymptotic regime.   However, in the finite regime, 
$\epsilon$ term appears as a small incremental gain in the EMI formulation, which needs to be considered.
Let ${\bf V}^*_{d}$  to be the unknown optimal spreading matrix that maximizes ${\mathrm J}({\bf V},\sigma^2)$ subject to the sparsity constraints. 
Also, let   ${\mathrm J}_{d}^*\triangleq{\mathrm J}({{\bf V}}_{d}^*,\sigma^2)$ to be the maximum of ${\mathrm J}({{\bf V}},\sigma^2)$ attained by ${\bf V}^*_{d}$.
Now, the penalty when using any spreading matrix $\bar {\bf V}^*\in \bar{\mathcal{C}}^*$, given as a solution of the system of equations in~\eqref{eq:system of underdetermined eq}, instead of the optimal one ${\bf V}^*_{d}$ can be written as
\begin{equation}
\Delta_{d}=\underbrace{{\mathrm J}({\bf V}_{d}^*,\sigma^2)}_{{\mathrm J}^*_{d}}-{\mathrm J}(\bar{\bf V}^*,\sigma^2).
\end{equation}
Writing ${\mathrm J}({\bf V}_{d}^*,\sigma^2)$ and  ${{\mathrm J}}(\bar{\bf V}^*,\sigma^2)$ as the deterministic equivalents plus the residual terms, we get
\begin{subequations}
	\begin{eqnarray}
	&{\mathrm J}({\bf V}_{d}^*,\sigma^2)-\bar{{\rm J}}({\bf V}_{d}^*,\sigma^2)=\epsilon^{(1)} \\
	& {{\mathrm J}}(\bar{\bf V}^*,\sigma^2)-\bar{{\rm J}}(\bar{\bf V}^*,\sigma^2)=\epsilon^{(2)}
	\end{eqnarray}
\end{subequations}
where the additional index $i$ in $\epsilon^{(i)}$ is added to distinguish the above differences. Subtracting  the sides of  the above equalities, and rearranging the terms,  we get
$$
\underbrace{({\mathrm J}({\bf V}_{d}^*,\sigma^2)-{\mathrm J}(\bar{\bf V}^*,\sigma^2)) }_{\Delta_{d}}+
(\bar{{\rm J}}(\bar{\bf V}^*,\sigma^2)-\bar{{\rm J}}({\bf V}_{d}^*,\sigma^2))=
\Delta_{\epsilon}
$$
where $\Delta_{\epsilon}=\epsilon^{(1)}-\epsilon^{(2)}$.
Since the subtraction $(\bar{{\rm J}}(\bar{\bf V}^*,\sigma^2)-\bar{{\rm J}}({\bf V}_{d}^*,\sigma^2))$ in the left-hand is positive\footnote{Note that $\bar{{\rm J}}(\bar{\bf V}^*,\sigma^2)$ is the maximum of the objective function in~\eqref{eq:Opt Asymp Rewriten}.}, it can be claimed that the gap to the optimum $\Delta_{d}$ is bounded as
\begin{equation}
\label{eq:delta ineqlity}
\begin{aligned}
0 \leq \Delta_{d} &\leq \Delta_{\epsilon}^+
\end{aligned}
\end{equation}
where $\Delta_{\epsilon}^+$ denotes the positive values of $\Delta_{\epsilon}$. In general, $\Delta_{d}$ may attain negative values. However, ${\mathrm J}^*_{d}<{\mathrm J}(\bar{\bf V}^*,\sigma^2)$ can happen only when $\bar{\bf V}^*$ violates the sparsity constraints, which is not the case of interest. Therefore, $\Delta_{d}$ is lower bounded by zero.
Next, we look into the properties of the residual terms in the finite regime to characterize the gap to the optimum as given in~\eqref{eq:delta ineqlity}.

Generally speaking, 
the residual term can be roughly associated with the  concentration of UEs' powers in a fewer number of the elements in the spreading codes. This gives rise to the variance of the random channel  entries, which eventually appears as  $\epsilon$ in the EMI formulation.
 In order to get further insight into the structure of $\epsilon$ term, a numerical example is illustrated in Fig.~\ref{fig:Epsil var mean}. The results are generated based on the simulation assumptions given in Section~\ref{sec:sim setting} with $F=50$, $K=100$, and $d_K=d, \forall k$.
  For generating the results in Fig.~\ref{fig:Epsil var mean}, a randomly selected drop of UEs is taken, and the mean and the variance of the residual term $\epsilon = {\mathrm J}({\bf V}_{d},\sigma^2)-\bar{{\rm J}}({\bf V}_{d},\sigma^2)$ are evaluated over 1000 randomly selected spreading matrices~${\bf V}_{d} \in \mathcal{C}_1$. The subscript $d$ is used to emphasize that the spreading codes in the columns of ${\bf V}_{d}$ have $d$ non-zero values. 
     \begin{figure}[h]
  	\centering
  	\includegraphics[width=0.6\linewidth]{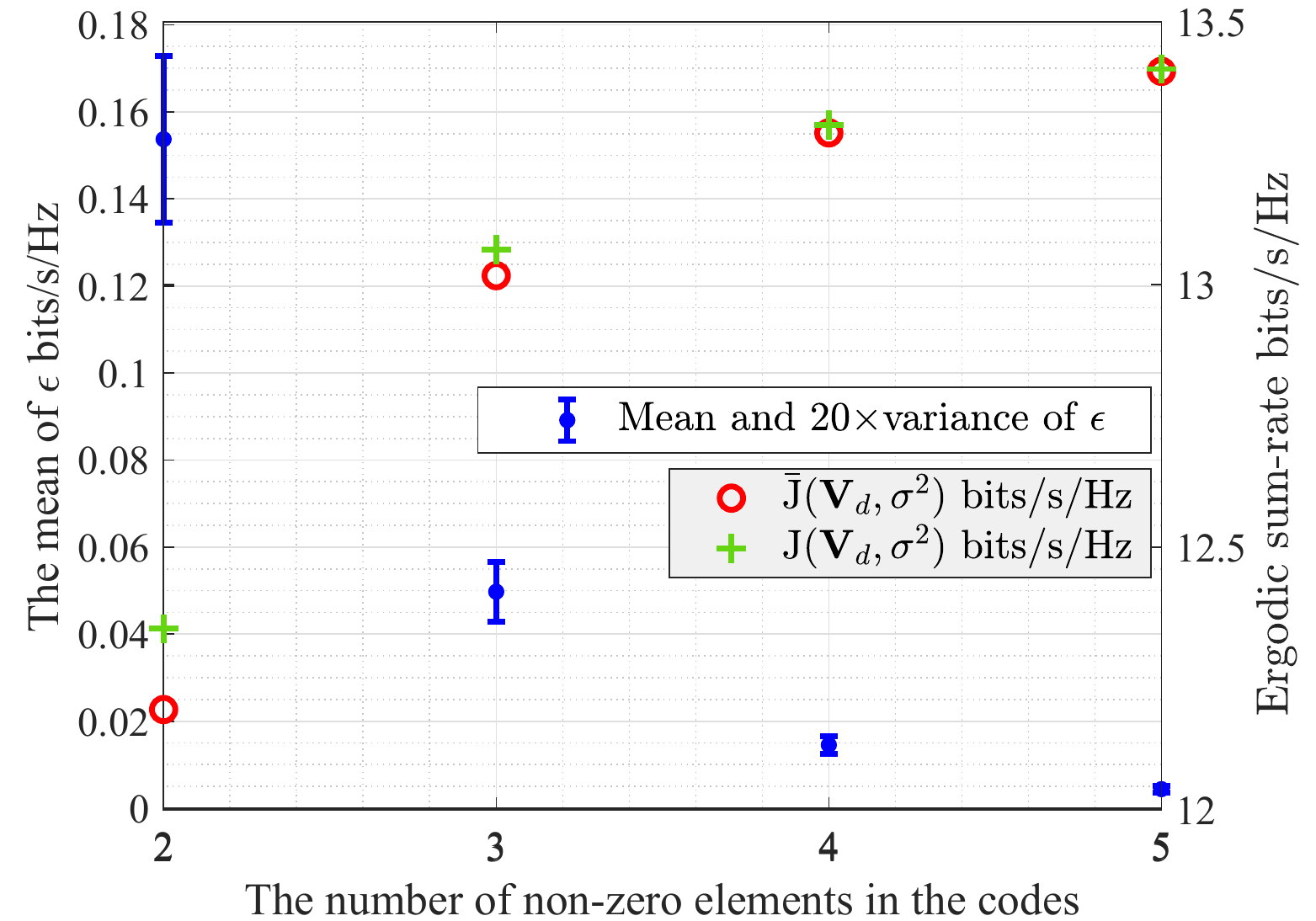}
  	\vspace{-0.3cm}
  	\captionof{figure}{Mean\&Var of $\epsilon$ (left axis), and ergodic sum-rate (right axis)  vs. $d$ for random spreading.}
  	\label{fig:Epsil var mean}
  \end{figure}
  Fig.~\ref{fig:Epsil var mean}   shows the mean and the variance of  $\epsilon$ values along with ${\mathrm J}({\bf V}_d,\sigma^2)$ and $\bar{\mathrm J}({\bf V}_d,\sigma^2)$ versus the number of non-zero elements  in the codes. The variance of $\epsilon$, depicted by the bars in the figure, is magnified 20 times for better illustration.
 The first observation is that the mean values of the $\epsilon$ term are distinct for different $d$ values. Moreover,
  the small variance of the residual term indicates that the values of $\epsilon$ do not vary abruptly among the power constrained  spreading vectors  with the same number of non-zero elements. Thus,  one might infer that the values of the residual term is dictated mainly by the number of non-zero elements in the spreading codes, which implies a decline in $\epsilon$ values as $d$ increases. The other observation  is that the residual term is small as compared to $\bar{\mathrm J}({\bf V}_d,\sigma^2)$.
  Generally speaking, in the scenarios with a moderate number of sub-channels $F$ as in Fig.~\ref{fig:Epsil var mean} , the residual term is larger for the sparse spreading  schemes with $d<<F$. While it almost disappears in the dense spreading case with $d=F$.
  This motivates the small incremental gain in $\epsilon$ for sparse spreading to be interpreted as sparsity gain. Note that there is a small loss associated with spreading in the considered system model.\footnote{ See~\cite{CDMAcodingTradeoff02} for the spreading coding trade-off.} Thus, the sparsity gain is coined here to reflect the increment in EMI due to sparse spreading as compared to  dense spreading.
  Recall the symmetric scenario mentioned in Section~\ref{sec: opt code asymp} wherein both the sparse-regular and the dense spreading matrices were among the solutions maximizing the deterministic $\bar{\mathrm J}({\bf V},\sigma^2)$ in~\eqref{eq:Opt Asymp Rewriten}. According to the above discussion, we expect that the exact EMI ${\mathrm J}({\bf V},\sigma^2)$ for the sparse-regular codes to be better of that in the  dense spreading scheme by an amount of $\epsilon$. Note that  $\epsilon$ for the dense spreading scheme is negligible. This has been also observed in~\cite{shitz2017} under the symmetric AWGN channel model. It is shown in~\cite{shitz2017} that the sparse-regular spreading codes yield slightly higher spectral efficiency as compared to the dense spreading scheme in the  symmetric scenario considered therein.

      To conclude this section, let us recall the $\Delta_{d}$ formulation in~\eqref{eq:delta ineqlity}, which gives the gap to the optimum ${\mathrm J}_{d}^*$ when using a solution $\bar{\bf V}^*\in \bar{\mathcal{C}}^*$ in the finite regime.
  Relying on the above analysis, we expect the $\epsilon^{(i)}, i=1,2$ terms to be small as compared to the corresponding deterministic parts. Also,   
  a solution $\bar {\bf V}^*$ that satisfies the sparsity constraints is expected to harness an incremental $\epsilon^{(2)}$ value close to the $\epsilon^{(1)}$ value. Therefore, the anticipated performance gap $\Delta_{d}$ is close to zero. Next, we find the desired sparse solutions in $\bar{\mathcal{C}}^*$ via an algorithmic solution.

\section{An Algorithm For Constructing The Sparse Spreading Matrices}
\label{sec:Alg max}
The problem of finding the subset of solutions in $\bar{\mathcal{C}}^*$ with the desired sparsity includes a zero-norm. The discrete and discontinuous nature of the zero-norm impedes the application of  standard convex analysis~\cite{BrucksteinSparseSol}. 
Fortunately, the system of equations in~\eqref{eq:system of underdetermined eq} unveils a simple rule for the allocation of the spreading codes. This allows the desired sparse codes to be obtained using an alternative algorithmic solution.
We are interested in determining the sparse spreading codes that satisfy~\eqref{eq:system of underdetermined eq}. 
Let the elements $v_{f,k}$ for each UE $k$ to be taken from $\{0,\frac{P_k}{d_k}\}$ and subject to the power constraint enforced by equalities in~\eqref{eq:system of underdetermined eq a}.
Based on~\eqref{eq:system of underdetermined eq b}, the problem is to allocate $\frac{P_k}{d_k}\beta_k$ values to sub-channels such that the sums of $\frac{P_k}{d_k}\beta_k$ values on each sub-channel become the same, i.e., equal to ($\frac{1}{r^*}-1$), across all the sub-channels.
 This problem falls within a class of partitioning problems that arises  in number theory and computer science~\cite{KORF1998181}.
Although the partitioning problem is NP-complete, there are heuristics that solve the problem in many instances, either optimally or approximately~\cite{EasyNPpoblm}.  
One such approach  is the greedy algorithm, which iterates through  $\frac{P_k}{d_k}\beta_k$ values in descending order, assigning each of them to whichever sub-channel has the smallest sum~\cite{Graham1969BoundsOM}.  These steps are summarized in Alg.~\ref{alg:partition}.
\begin{algorithm} [H]
	\caption{Partitioning solution}
	\label{alg:partition}
	\begin{algorithmic}[1]
		\STATE Divide the total power of each UE $k$ into $d_k$ power fragments.
		\STATE Set $v_{f,k}=0,\,\forall f,k$, and $\mathcal{J}=\{1,...,K\}$.
		\WHILE{$\mathcal{J}$ is non-empty}
		\STATE Set $k=\underset{j\in\mathcal{J}}{\arg\max} \,\frac{P_j}{d_j}\beta_j $.
		\STATE Set $f=\underset{i\in\mathcal{F}}{\arg\min} \,\eta_i$ with $\eta_i=\sum_{j\in\mathcal{K}} {\beta_j v_{i,j}}$.
		\STATE Set $v_{f,k}=P_k/d_k$ 
		\IF { $\sum_{i\in\mathcal{F}} v_{i,k}=P_k$, i.e., UE $k$ has assigned all of its power-fragments}
		\STATE Remove index $k$ from $\mathcal{J}$.
		\ENDIF		
		\ENDWHILE
	\end{algorithmic}
\end{algorithm}
In Alg.~\ref{alg:partition}, we try to make the sum terms $\eta_f$ across the sub-channels as equal as possible.
Let  ${\eta^*_{\rm{max}}}$ denotes the maximum of $\eta_f, \forall f\in\mathcal{F}$ in an optimal partitioning solution. Alg.~\ref{alg:partition} yields $\eta_f$ values such that
$\frac{\rm{max}(\eta_f)}{\eta^*_{\rm{max}}}\leq \frac{4}{3}-\frac{1}{3F}$ ~\cite{Graham1969BoundsOM}. One can always improve the homogeneity of $\eta_f$ values  by offloading the power from the sub-channels with largest  $\eta_f$ values into those with corresponding smallest values until~\eqref{eq:system of underdetermined eq b} holds up to the desired accuracy.
However, numerical analysis shows that Alg.~\ref{alg:partition} yields satisfactory results in most cases, and thus, further fine-tuning steps are redundant.
 Alg.~\ref{alg:partition} has a running time of $\mathcal{O}(2^F (d_{\rm{max}} K)^2)$~\cite{Graham1969BoundsOM}.

\section{Numerical results}
\label{sec:num res}
\subsection{Simulation assumptions}
\label{sec:sim setting}
The simulation results are generated in a scenario where a single-antenna BS serves single-antenna UEs in uplink. Transmit power of each UE is $1$ Watt, and the noise power is set to $-120$dB. The pathloss values are taken randomly and uniformly from the range of $-150$dB to $-60$dB to account for diverse received SNRs at the BS. 
 The final channel gains are given by the product of the pathlosses and the small-scale fading entries as in~\eqref{eq:resive vec}.
To keep the results comparable, the numbers of non-zero elements in the spreading codes are assumed to be the same for all UEs, i.e., $d_k=d,\forall k$.
The number of sub-channels is $F=50$ unless~mentioned~otherwise.

In addition to the method proposed in Section~\ref{sec:Alg max},  two alternative spreading schemes from the literature are considered.
The first one is the random spreading method that allocates the power constrained sparse spreading codes to UEs randomly.  A practically useful property of the random spreading scheme is that no coordination overhead is imposed. The other scheme is the coordinated regular spreading, which assigns the non-zero elements in the codes in a way that  each UE occupies a number of $d$ resources and each resource is used by a number of $\frac{K}{F}d$ UEs.
The ratio $\frac{K}{F}$ is chosen such that $\frac{K}{F}d$ be an integer.
In evaluating the performance of the considered methods, we average the corresponding attainable rates over 1000 UE drops, where, in each drop, the expectation involved in EMI ${\rm J}({\bf V},\sigma^2)$
is evaluated over 1000 random realizations of small-scale fading. The deterministic equivalent   values $\bar{\rm J}({\bf V},\sigma^2)$ are evaluated~from~\eqref{eq:asymp I theorem}.

\subsection{Evaluation of the performance of the proposed method}
      \begin{figure}[h]
	\centering
	\includegraphics[width=0.7\linewidth]{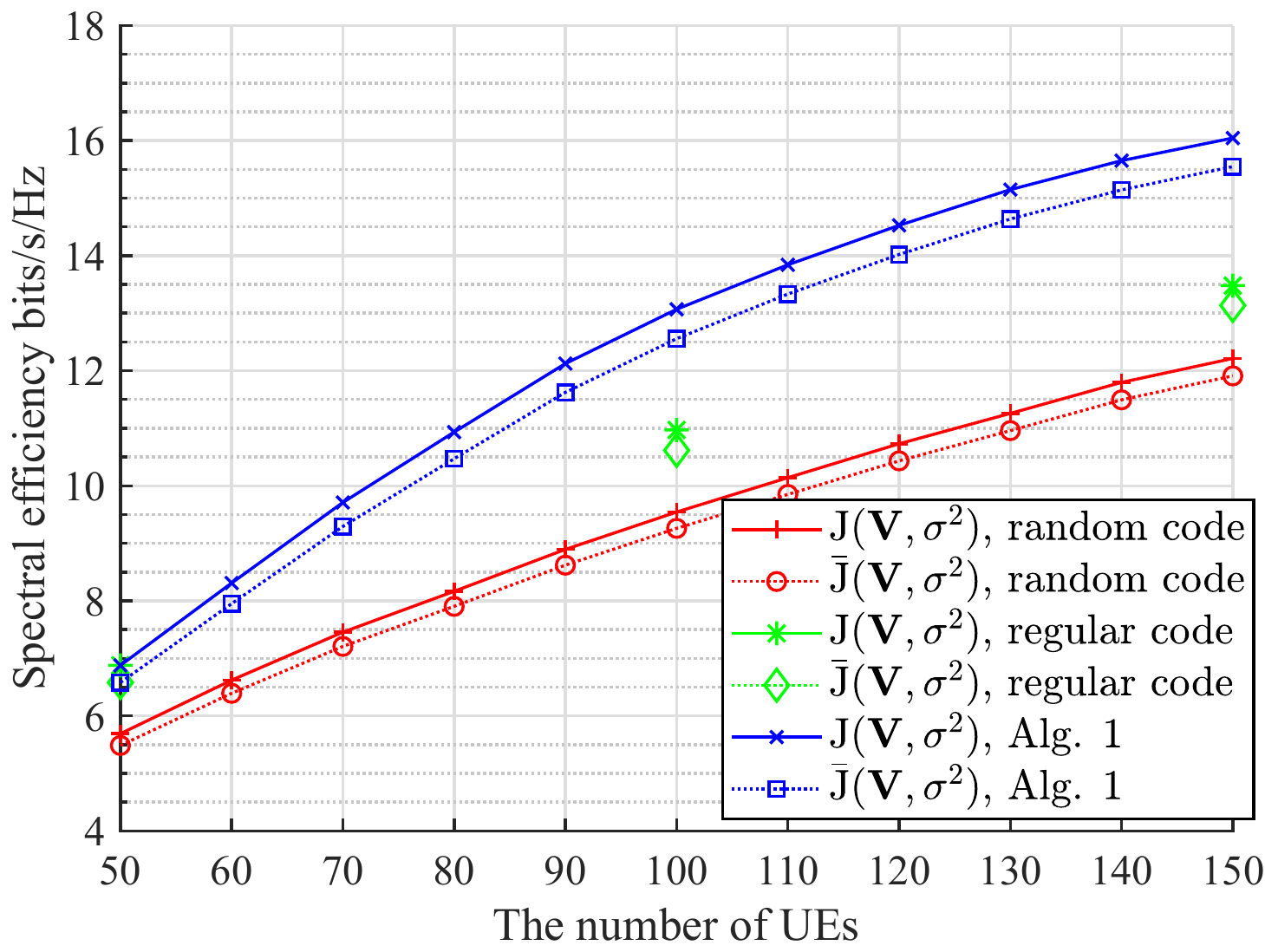}
	\captionof{figure}{The spectral efficiency ${\rm J}({\bf V},\sigma^2)$ vs. the number of UEs,\,$d=1, F=50$.}
	\label{fig:rate vs K d1}
\end{figure}
 The assessment of an optimal solution to the problem in~\eqref{eq:optimizeMutualInfo} requires exhaustive search over all the power constrained spreading matrices with the desired sparsity. This impedes the comparison of the results to the optimum. However, still we can evaluate and compare the performance in Alg.~\ref{alg:partition} 
 and in the aforementioned uncoordinated random and the coordinated regular spreading schemes.
Note that while the uncoordinated assignment of the spreading codes is a useful property, the random spreading causes a significant performance degradation. On the other hand, the coordinated allocation of spreading codes in the regular scheme does not consider the asymmetry in the system model due to the diverse pathloss values and power constraints at the UEs. On the contrary, the spreading codes in Alg.~\ref{alg:partition}
are allocated according to the UEs' pathloss values. This results in the minimal coordination requirement since the statistical properties of the channel matrix can be assumed to remain constant for a sufficiently large number of reception phases~\cite{BjornsonLuca2016}.

Figures~\ref{fig:rate vs K d1} and~\ref{fig:rate vs K d2} show the attainable rates (bits/s/Hz) for the aforementioned methods as a function of the number of UEs. 
In Fig.~\ref{fig:rate vs K d1}, we apply Alg.\ref{alg:partition} to the special non-spreading case with $d=1$ as well.
Note that Alg.\ref{alg:partition} allocates the sparse codes to UEs such that the deterministic equivalent of EMI is maximized. The motivation therein is that the residual term  is small relative to the deterministic equivalent part, and the small gain in $\epsilon$ is harnessed inherently due to the sparsity of the allocated codes.  
While we expect $\epsilon$ to be relatively small for the cases with $d>1$, due to the fast convergence rate of $\mathcal{O}(\frac{1}{d^2})$, the analysis in the non-spreading case with $d=1$ may be considered as a heuristic attempt.
 Interestingly, the difference between ${\rm J}({\bf V},\sigma^2)$ and $\bar{\rm J}({\bf V},\sigma^2)$ is relatively small even in such a case, and the coordinated allocation of resources gives 20\% and 35\% enhancement in the spectral efficiency at 100\% and 300\% system load, respectively. The system load is defined as the ratio of $\frac{K}{F}$ in percentage.
Fig.~\ref{fig:rate vs K d2} shows the rates for the case with LDS codes having $d=2$ non-zero values. In this case, the gain in the coordinated assignment of spreading codes is about 6.5\% and 11\% at 100\% and 300\% load, respectively, which is less than that in Fig.~\ref{fig:rate vs K d1}. The other observation is that the regular spreading method in both of Figs.~\ref{fig:rate vs K d1} and~\ref{fig:rate vs K d2} gives slightly better spectral efficiency as compared to the random spreading scheme, however, its performance 
is inferior to that of Alg.~\ref{alg:partition}.
\begin{figure}[t!]
	\centering
	\includegraphics[width=0.7\linewidth]{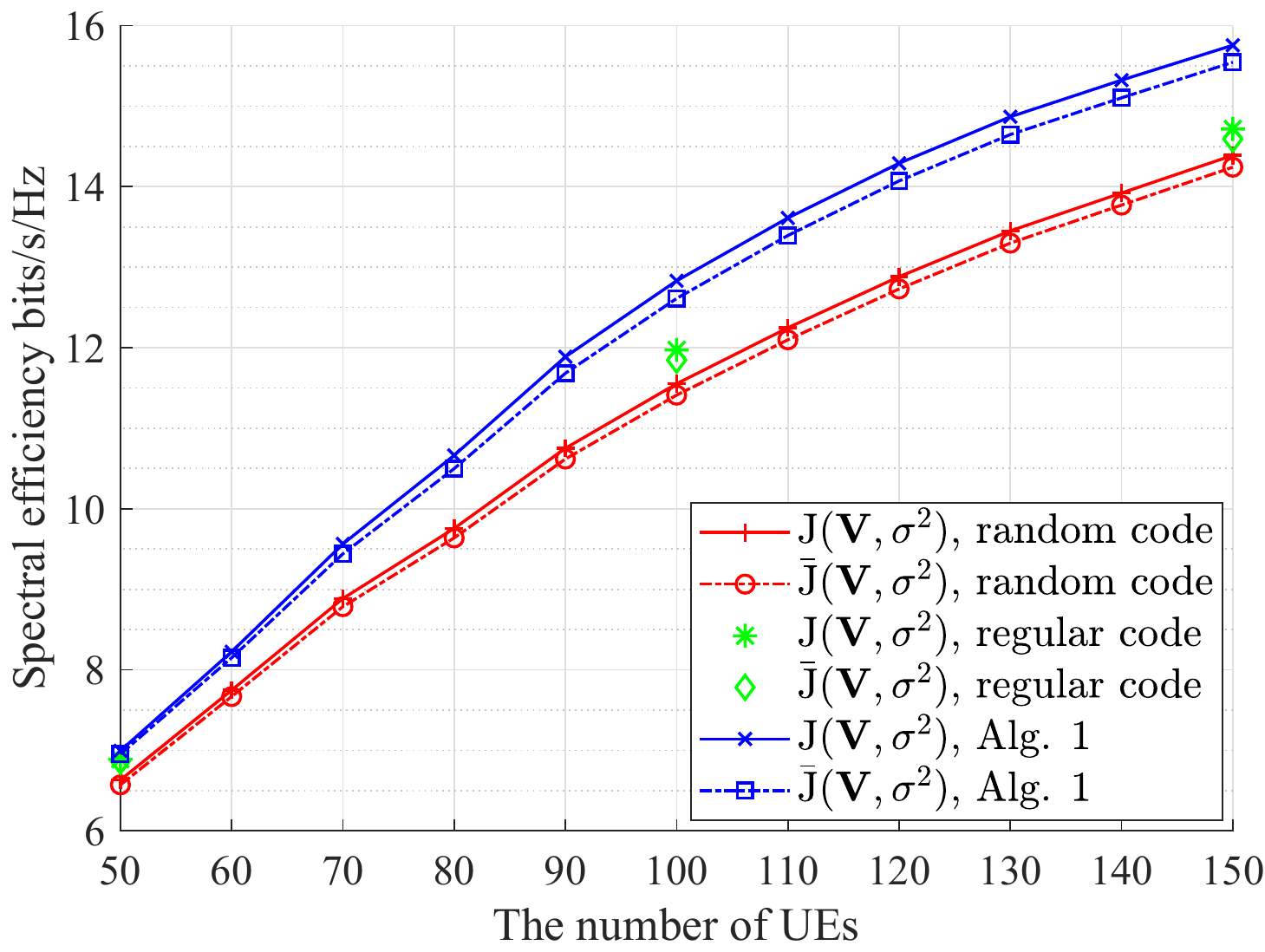}
	\captionof{figure}{The spectral efficiency ${\rm J}({\bf V},\sigma^2)$ vs. the number of UEs,\,$d=2, F=50$.}
	\label{fig:rate vs K d2}
\end{figure}

 In Fig.~\ref{fig:rate vs d}, the attainable SEs are presented versus the number of non-zero elements $d$, for a system load of 300\%. 
It can be seen that 
the performance of the random spreading method  improves as $d$ grows larger.
Spreading on more sub-channels allows UEs to attain interference diversity. This, in general, reduces the loss imposed by the uncoordinated resource allocation.
Note that, even though one can enhance the spectral efficiency of the uncoordinated method by spreading on further sub-channels, the number of UEs overlapping on the same sub-channel increases correspondingly. In a system with 300\% load, the average number of overlapping UEs in the case with $d=2$ and $d=6$ is equal to $6$ and $18$, respectively. Thus, the detection complexity is greatly increased with larger values of $d$. In Fig.~\ref{fig:rate vs d}, the spectral efficiency of the dense spreading scheme is also depicted. As mentioned in Section~\ref{sec: opt code asymp}, the dense spreading matrix is a solution of the optimization problem in~\eqref{eq:Opt Asymp Rewriten}, and thus,  the values of deterministic $\bar{\rm J}({\bf V},\sigma^2)$ for both the dense spreading scheme and the sparse spreading in Alg.~\ref{alg:partition} are the same. However, the values of the exact EMI ${\rm J}({\bf V},\sigma^2)$ for the sparse spreading in Alg.~\ref{alg:partition} are better of that in the dense spreading scheme by an amount of $\epsilon$. This can be seen from Fig.~\ref{fig:rate vs d} where the curves of $\bar{\rm J}({\bf V},\sigma^2)$ and ${\rm J}({\bf V},\sigma^2)$ are almost overlapping for the dense spreading  case, while the value of ${\rm J}({\bf V},\sigma^2)$ for Alg.~\ref{alg:partition} at $d=2$ is nearly 0.2 bits/s/Hz higher than the deterministic EMI. This additional gain in  ${\rm J}({\bf V},\sigma^2)$ for sparse spreading as compared to dense spreading was referred to as sparsity gain in Section~\ref{sec:asymp in finit}. Observe that the sparsity gain decreases as the number of non-zero elements in the codes increases. Finally, we observe that, in contrast to the symmetric model\footnote{See Section~\ref{sec:asymp in finit} for definition of symmetric system model}, the regular spreading method is inferior to Alg.\ref{alg:partition} and even to the dense spreading scheme
 in the considered asymmetric scenario. As mentioned in Section~\ref{sec:asymp in finit}, the regular spreading matrix is an optimal solution to the optimization problem in~\eqref{eq:Opt Asymp Rewriten} in the symmetric scenarios.
    \begin{figure}[t!]
	\centering
	\includegraphics[width=0.7\linewidth]{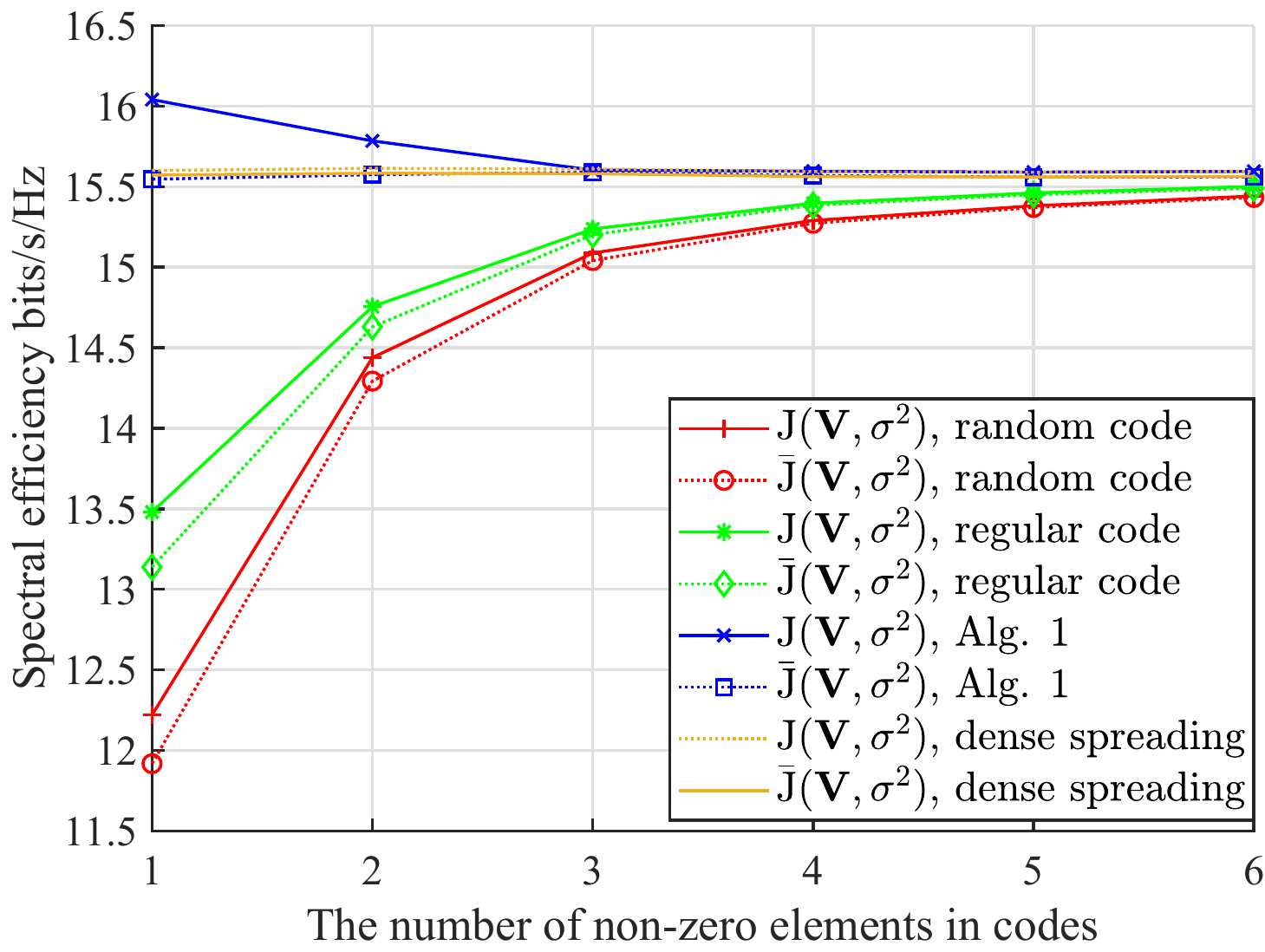}
	\vspace{-0.3cm}
	\captionof{figure}{The spectral efficiency ${\rm J}({\bf V},\sigma^2)$ vs.  $d$, $F=50$, $\frac{K}{F}=3$.}
	\label{fig:rate vs d}
\end{figure}

\subsection{Visualization of the resource allocation strategy in the proposed method}
In Fig.~\ref{fig:alloc}, the resource allocation strategy in Alg.~\ref{alg:partition} 
is visualized in a scenario with $F=30$, $K=90$, and $d=2$. This figure illustrates the allocation of the spreading codes in Alg.~\ref{alg:partition} 
for a particular drop of UEs. The colorbar represents the unitless $\beta_k$ values, introduced in~\eqref{eq:system of underdetermined eq}.
It can be seen that the power fragments are allocated to the sub-channels such that the sum terms $\eta_f=\sum_{k\in\mathcal{K}} {\beta_k v_{f,k}}$ become equal. Observe also that the UEs with large $\beta_k$ values are distributed across the sub-carriers. Then, those with smaller $\beta_k$ values are placed such that the sum terms $\eta_f=\sum_{k\in\mathcal{K}} {\beta_k v_{f,k}}$  become equalized. 
It can be seen from Fig.~\ref{fig:alloc} that the UEs overlapping on each sub-carriers have diverse $\beta_k$ values. 

   \begin{figure}[t!]
	\centering
	\includegraphics[width=0.6\linewidth]{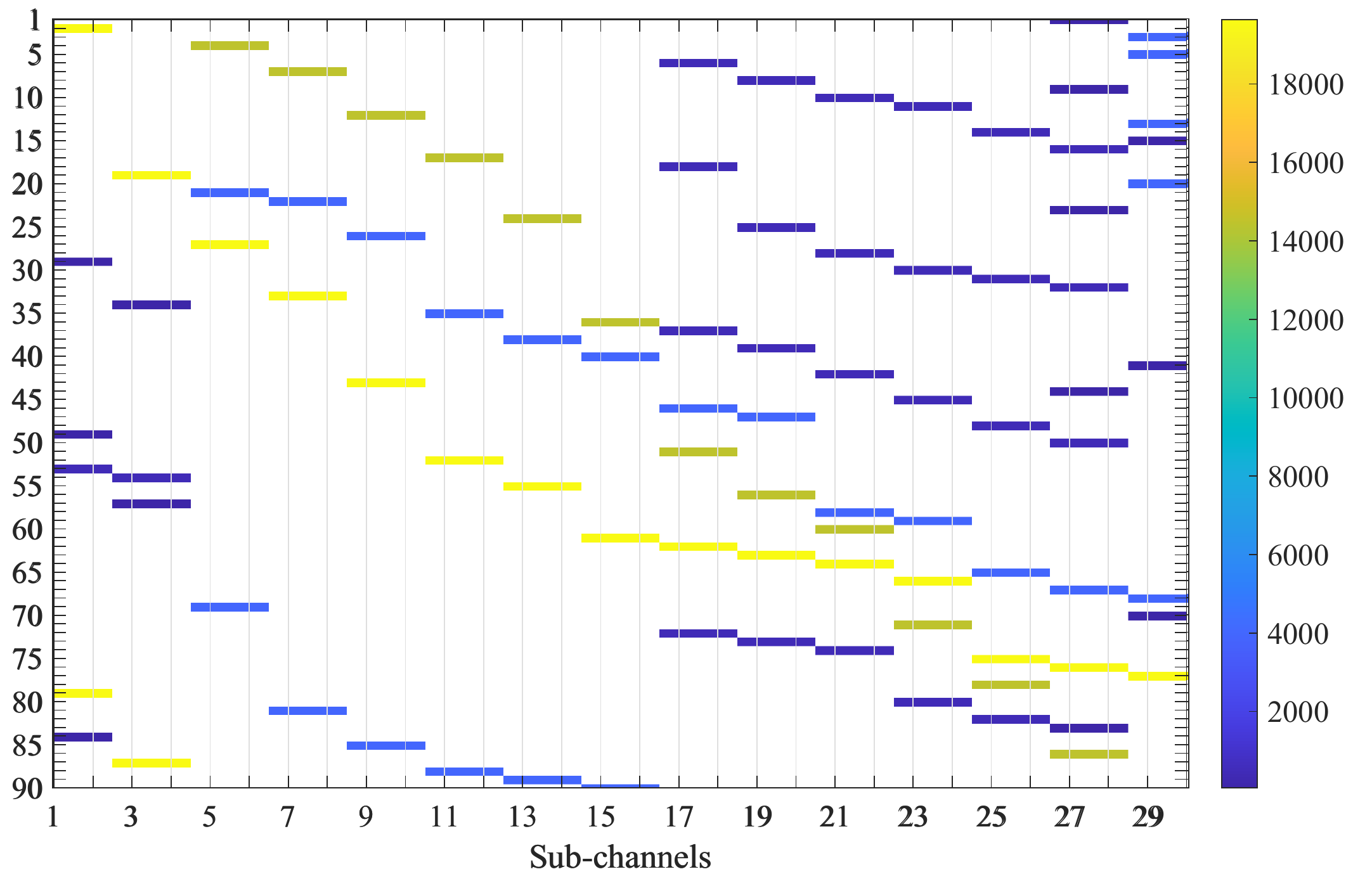}
	\captionof{figure}{Visualization of the resource allocation strategy in Alg.~\ref{alg:partition}, $F=30$, $K=90$, $d=2$.}
	\label{fig:alloc}
\end{figure}

\section{Conclusions}
\label{sec:Conclusion}
In this paper, a simple and efficient rule for close-to-optimal allocation of sparse spreading codes was derived based on rigorous analysis.
   The analysis reduced the dilemma of maximizing the ergodic mutual information to a partitioning problem, which was solved via an efficient algorithmic solution.
The proposed algorithm  allocates the spreading codes based on the system load, the sparsity constraints and pathloss values. 
The simulation results showed that the proposed algorithm with minimal coordination provides a superior performance as compared to the uncoordinated random spreading and the coordinated regular spreading schemes. It was shown that the regular spreading matrices are asymptotically optimal only in the symmetric system model, while in the asymmetric case the performance of the regular spreading method is even inferior to the dense spreading scheme.
As the future work, we are interested in extending the results to the multi-antenna BS scenario with  correlation introduced among the antenna elements and among the sub-carriers. Following the same optimization approach as in here, we expect  
the optimal low density spreading policies to be characterized based on the spatial and the spectral  correlation properties of the channel~matrix. 
\appendices

\section{Proof of Theorem~\ref{th:cap conv}}
\label{sec:prof of the 1}
The ergodic mutual information ${\mathrm J}_F ({\bf W}_F,\sigma^2)$ is related to $\frac{1}{F} \mathbb{E} \tr {(\frac{1}{\sigma^2}{\bf H}_{F} {\bf H}_{F} \herm+{\bf I}_{F})}^{-1}$ as follows
\begin{equation}
\label{eq:I deriv Tr}
\frac{\partial {\mathrm J}_F ({\bf W}_F,\sigma^2)}{\partial \sigma^2}=\frac{1}{\sigma^2F}\mathbb{E}{\mathrm{Tr}}(\frac{1}{\sigma^2}{\bf H}_{F} {\bf H}_{F} \herm+ {\bf I}_{F})^{-1}-\frac{1}{\sigma^2}.
\end{equation}
Denoting ${\bf{Q}}_{F}(\sigma^2)=(\frac{1}{\sigma^2}{{\mathbf{H}}_{F}} {{\bf H}_{F}}\herm+\mathbf{I}_{F})^{-1}$, equivalently we have 
\begin{equation}
\label{eq:intg tr}
{\mathrm J}_{F} ({\bf W}_F,\sigma^2)=\int_{\sigma^2}^{+\infty}\big( \frac{1}{ z}-\frac{1}{zF}\mathbbm{E}{\tr} {\bf Q}_{F}(z) \big)dz.
\end{equation}
Thus, given an expression for the trace term $\mathbbm{E}{\tr}{\bf Q}_{F}(z), \forall z\in\mathbb{R}^+$, one can equivalently derive an expression for EMI based on~\eqref{eq:intg tr}. 
In derivation of the results, we utilize so-called Gaussian method~\cite{RMT,Hachem07}.
Let ${\bf x} \sim \mathcal{CN}({\bf 0},{\sbf \Theta})$ be a circularly symmetric Gaussian random vector with covariance matrix ${\sbf \Theta}\in\mathbb{C}^{N\times N}$. Also, let the  function $f({\bf x},{\bf x}^*)\in \mathbb{C}$ together with its derivatives be polynomial bounded.
The Gaussian method consists of two ingredients~\cite{Hachem07}: 
\begin{itemize}
	\item Integration by parts formula: $
	\mathbb{E}x_j f({\bf x})= \sum_{i=1}^{N}[{\sbf \Theta}]_{j,i}  \mathbb{E}\frac{\partial f({\bf x})}{\partial x_i^*}$,
	\item Nash-Poincar\'{e} inequality: $
	{\var}(f({\bf x}))\leq \mathbb{E}\nabla_{\bf x}f({\bf x})\tran{\sbf \Theta} \big(\nabla_{\bf x}f({\bf x})\big)^*+ 
	\mathbb{E}\big(\nabla_{{\bf x}^*}f({\bf x})\big)\herm{\sbf \Theta} \nabla_{{\bf x}^*}f({\bf x})$.
\end{itemize}
In the sequel, we show that  a diagonal element $q_{p,p}$ of matrix ${\bf{Q}}_F$ can be written as the sum of  ${h}_{i,j} q_{p,i}  {h}_{p,j}^*$ terms.
Defining the function $f({\bf H}_F)\triangleq q_{p,i}{h}_{p,j}^*$, then, we expand the terms of  type $\mathbbm{E} {h}_{i,j}f({\bf H}_F)$ using the Gaussian integration by parts formula. This retrieves an
implicit but deterministic expression for $\mathbbm{E}\tr({\bf Q}_{F})$ up to a small residual term. Then, Nash-Poincar\'{e} inequality allows us to derive an upper bound on the residual term and declares the convergence of the trace term under the realms of Theorem~\ref{th:cap conv} as 
\begin{equation}
\label{eq:tr q in proof}
\frac{1}{F}\mathbbm{E}{\tr} {\bf Q}_{F}(\sigma^2)=\frac{1}{F}{\tr}{\bf R}_{F}({\bf W}_F,\sigma^2)+\zeta_{F}(\sigma^2)
\end{equation}
where ${\bf R}_{F}({\bf W}_F,\sigma^2)$ is as defined in the theorem, and $\zeta_F(\sigma^2)=\mathcal{O}(\frac{1}{d^2})$ is a fast diminishing term.
Given the deterministic equivalent for $\frac{1}{F}\mathbbm{E}{\tr} {\bf Q}_{F}(\sigma^2)$ as in~\eqref{eq:tr q in proof}, the convergence of EMI ${\mathrm J}_{F} ({\bf W}_F,\sigma^2)$ to the deterministic equivalent $\bar{\mathrm J}_{F} ({\bf W}_F,\sigma^2)$, as declared in Theorem~\ref{th:cap conv}, follows from the relation in~\eqref{eq:intg tr}. In the following, we proof the convergence of the trace term in~\eqref{eq:tr q in proof}. The proof of the convergence of EMI is straightforward, and thus, is omitted. We invite the reader to refer to~\cite[Theorem 1]{Hachem07} for the details. In the sequel, we frequently omit the subscript $F$ denoting the dependency of the entities on the system size.
In proving the convergence of the trace term, we frequently use the following elementary results,
\begin{align}
\label{eq:deriv Q}
\frac{\partial q_{p,d}}{\partial {h}_{i,j}}&=-[{\bf Q}( \frac{\partial\bf Q^{-1}}{\partial {h}_{i,j}}  ){\bf Q}]_{p,d}=-\frac{1}{\sigma^2}[{\bf Q}( \frac{\partial{{\bf H}{\bf H}}\herm}{\partial {h}_{i,j}}  ){\bf Q}]_{p,d}\\
&=-\frac{1}{\sigma^2} q_{p,i} [ { {\bf h}_j\herm \bf Q}]_{d},
\end{align}
and similarly,
\begin{align}
\label{eq:deriv Q tild}
\frac{\partial q_{p,d}}{\partial {h}_{i,j}^*}=-\frac{1}{\sigma^2}[ {\bf Q} {\bf h}_j]_{p}\, q_{i,d}.
\end{align}

We start by noticing that ${\bf Q}={\bf I}_F-\frac{1}{\sigma^2}{\bf Q} {\bf H} {\bf H}\herm$, a relation often referred to as the
resolvent identity. This allows one to write $\mathbbm{E}q_{p,p}$ as a function of $\mathbbm{E}[{\bf Q} {\bf H} {\bf H}\herm]_{p,p}$, i.e.,
\begin{equation}
\label{eq:resolvant id}
\mathbbm{E}{q}_{p,p}=1-\frac{1}{\sigma^2}\sum_{j=1}^{K}\mathbbm{E}[ {\bf Q} {\bf h}_j]_{p} {h}_{p,j}^*,\,\, p\in \mathcal{F}.
\end{equation}
Now, we work on term $\mathbbm{E}[ {\bf Q} {\bf h}_j]_{p} {h}_{p,j}^*=\sum_{i=1}^{F}\mathbbm{E}\, q_{p,i} {h}_{i,j} {h}_{p,j}^*$  to expand it using the Gaussian integration by parts formula. Let the function $f({\bf H})$ to be defined as $f({\bf H})\triangleq q_{p,i}{h}_{p,j}^*$. Then from the integration by parts formula, we get 
 \begin{subequations}
 	\label{eq:integ by part eqp term}
	\begin{align}
\mathbbm{E} {h}_{i,j}q_{p,i}{h}_{p,j}^*&=\mathbbm{E} {h}_{i,j}f({\bf H})= a_{j}^2 v_{i,j} \mathbbm{E}\frac{\partial f({\bf H})}{\partial {h}_{i,j}^*}\\
&=a_{j}^2 v_{i,j} \mathbbm{E}(\frac{\partial q_{p,i}}{\partial {h}_{i,j}^*} {h}_{p,j}^*+\frac{\partial{{h}}_{p,j}^* }{\partial {h}_{i,j}^*} q_{p,i})\\
&=-\frac{1}{\sigma^2}a_{j}^2 v_{i,j} \mathbbm{E}[ {\bf Q} {\bf h}_j]_{p}\, q_{i,i} {h}_{p,j}^*+a_{j}^2 v_{i,j} \delta(i-p)\mathbbm{E}q_{p,i}
\end{align}
\end{subequations}
where ${v}_{i,j}\triangleq \frac{1}{d_{j}} w_{i,j}^2$. Summing the sides of above equality over index $i$, we get
\begin{equation}
\mathbbm{E}[ {\bf Q} {\bf h}_j]_{p} {h}_{p,j}^*=
-\frac{1}{\sigma^2} \mathbbm{E}[ {\bf Q} {\bf h}_j]_{p} {h}_{p,j}^*   {a_{j}^2} {\tr}( {\bf Q} {\bf{V}}_{j} )
+a_{j}^2 v_{p,j}\mathbbm{E}q_{p,p}
\end{equation}
where  ${\bf V}_j=\diag\{{\bf v}_{j} \}$ with ${\bf v}_l$ defined as ${\bf v}_l\triangleq[{ v}_{1,l},...,{ v}_{F,l}],\,\forall l$.
Let us define $\beta_{j}={a_{j}^2} {\tr}( {\bf Q} {\bf V}_{j} )$, $\alpha_j=\mathbbm{E}\beta_j$ and $\overset{o}{\beta_j}=\beta_j-\alpha_j$. Then, the terms in the above equation can be separated~as
\begin{equation}
\mathbbm{E}[ {\bf Q} {\bf h}_j]_{p} {h}_{p,j}^*=
-\frac{1}{\sigma^2}  \alpha_j \mathbbm{E}[ {\bf Q} {\bf h}_j]_{p} {h}_{p,j}^*  
+a_{j}^2 v_{p,j}\mathbbm{E}q_{p,p} -\frac{1}{\sigma^2} \mathbbm{E}[ {\bf Q} {\bf h}_j]_{p} {h}_{p,j}^* \overset{o}{\beta_j}.
\end{equation}
Solving the above equality for $\mathbbm{E}[ {\bf Q} {\bf h}_j]_{p} {h}_{p,j}^*$, we get
\begin{equation}
\label{eq: tilde e formula}
\mathbbm{E}[ {\bf Q} {\bf h}_j]_{p} {h}_{p,j}^*=\tilde{e}_j a_{j}^2 v_{p,j}\mathbbm{E}q_{p,p}-\tilde{e}_j \frac{1}{\sigma^2}  \mathbbm{E}[ {\bf Q} {\bf h}_j]_{p} {h}_{p,j}^* \overset{o}{\beta_j}
\end{equation}
where $\tilde{e}_j\triangleq \frac{\sigma^2}{\sigma^2+{\alpha_j}}$. Summing the sides of above equality over index $j$, we get
\begin{equation}
\label{eq:QHH expanded}
\mathbbm{E}[ {\bf Q} {\bf H}{\bf H}\herm]_{p,p} ={\tr}({\sbf{ \Lambda}}_{p}{\bf A}\tilde{\bf E})\mathbbm{E}q_{p,p}-\frac{1}{\sigma^2} \sum_{j=1}^{K}\tilde{e}_j \mathbbm{E}  [ {\bf Q} {\bf h}_j]_{p} {h}_{p,j}^* \overset{o}{\beta_j}
\end{equation}
where ${\sbf \Lambda}_{p}=\diag\{v_{p,_1},...,v_{p,_K}\}$, ${\bf A}=\diag\{a_1^2,...,a_K^2\}$, $\tilde{\bf E}=\diag\{\tilde{e}_1,...,\tilde{e}_{K}\}$,  and ${\overset{o}{\boldsymbol{\beta}}}=\rm{diag}\{\overset{o}{\beta_1},...,\overset{o}{\beta_{K}}\}$. Utilizing the resolvent identity,~\eqref{eq:resolvant id} and~\eqref{eq:QHH expanded} yield
\begin{equation}
\label{eq:qpp eq1}
{\sigma^2-\sigma^2\mathbbm{E}q_{p,p}}=\tilde{\alpha}_p\mathbbm{E}q_{p,p}-{\frac{1}{\sigma^2}}\sum_{j=1}^{K}\tilde{e}_j \mathbbm{E}  [ {\bf Q} {\bf h}_j]_{p} {h}_{p,j}^* \overset{o}{\beta_j}
\end{equation}
where $\tilde{\alpha}_p={\tr}({\sbf \Lambda}_{p}{\bf A}\tilde{\bf E})$. Now, one can solve~\eqref{eq:qpp eq1} to obtain $\mathbbm{E}q_{p,p}$. Let $\bf T$ be a $F\times F$ diagonal non-negative matrix with bounded spectral norm. Multiplying the acquired $\mathbbm{E}q_{p,p}$ from~\eqref{eq:qpp eq1} by elements of $\bf T$,  and summing  over $p$ yields
\begin{equation}
\label{eq:TrDQ plus var}
\frac{1}{F} \tr({\bf T}\bf{ Q})=\frac{1}{F}\tr({\bf T}\bf{ E})+{\frac{1}{\sigma^4 F}\sum_{j=1}^{K} \tilde{e}_j\mathbbm{E}   {\bf h}_{j}\herm {\bf T}{\bf E} {\bf Q} {\bf h}_j\overset{o}{\beta_j}}
\end{equation} 
where ${\bf E}=\diag\{{e}_1,...,{e}_{K}\}$ with $e_p\triangleq \frac{\sigma^2}{\sigma^2+\tilde{\alpha}_p}$. Next in Appendix~\ref{sec:convergence analys}, we prove that the last term in the above equation, hereafter denoted by $\zeta(\sigma^2)$,  vanishes with $\mathcal{O}(\frac{1}{d^2})$ rate. As the result, we get the convergence $\frac{1}{F} \tr({\bf T}\bf{ Q})-\frac{1}{F}\tr({\bf T}\bf{ E})\rightarrow 0$. However, notice that the term $\tr({\bf T} \bf{E})$ still depends on the unknown parameters $\alpha_{j}={a_{j}^2}\mathbb{E} {\tr}( {\bf V}_{j} {\bf Q} )$. Therefore, in the last step of the proof given in Appendix~\ref{sec:approx of E}, we need to show that the matrix $\bf E$ can be replaced by the deterministic matrix $\bf R$, as defined in the theorem.

\subsection{The upper-bound on $\zeta(\sigma^2)$ term}
\label{sec:convergence analys}
 In proving the convergence rate, we first derive an upper-bound for $|\zeta(\sigma^2)|$ in terms of the variance of $\tr(\vec{Q})$. Then, we show that the upper-bound vanishes with $\mathcal{O}(\frac{1}{d^2})$ rate. observing that $\beta_{j}={a_{j}^2} {\tr}( {\bf V}_{j}{\bf Q}  )$, the $\zeta(\sigma^2)$ term, given as the last term in~\eqref{eq:TrDQ plus var}, can be presented as in the following
 \begin{subequations}
 	\label{eq:zeta expand}
 	\begin{align}
 	|\zeta(\sigma^2)|&= \frac{1}{\sigma^4 F}\bigg|\sum_{i=1}^{F}\mathbbm{E}\, \overset{o}{q}_{i,i} \sum_{j=1}^{K} v_{i,j} a_j^4  \tilde{e}_j {\bf h}_{j}\herm {\bf T}{\bf E} {\bf Q} {\bf h}_j  \bigg|\\\label{eq:zeta expand1}
 	&\leq \frac{1}{\sigma^4 F}\sum_{i=1}^{F}\bigg|\mathbbm{E}\, \overset{o}{q}_{i,i} \sum_{j=1}^{K} v_{i,j} a_j^4  \tilde{e}_j {\bf h}_{j}\herm {\bf T}{\bf E} {\bf Q} {\bf h}_j  \bigg|\\\label{eq:zeta expand2}
 	 	&\leq \frac{1}{\sigma^4 d F}\sum_{i=1}^{F}\bigg|\mathbbm{E}\|\tilde{\vec{E}}\|\|\vec{A}\|^2\, \overset{o}{q}_{i,i} \sum_{j=1}^{K} {\bf h}_{j}\herm {\bf T}{\bf E} {\bf Q} {\bf h}_j  \bigg| 
 	\end{align}
 \end{subequations}
where $\overset{o}{q}_{i,i}={q}_{i,i}-\mathbbm{E}{q}_{i,i}$, $\vec{A}=\diag\{a_1^2,...,a_K^2\}$, and~\eqref{eq:zeta expand1} is due to the triangle inequality. Let $\phi\triangleq\tr(\vec{Q})$, then, using the resolvent identity we get
 \begin{subequations}
	\label{eq:Second zeta expand}
	\begin{align}
	|\zeta(\sigma^2)|&\leq \frac{C}{d F}\sum_{i=1}^{F}\bigg|\mathbbm{E}\big(\, \overset{o}{q}_{i,i} \tr\big(\vec{TE}(\vec{I}_F-\vec{Q})\big) \big) \bigg|\label{eq:Second zeta expand1}\\
        &\leq \frac{C}{ d F}\var(\phi)\label{eq:Second zeta expand3}
	\end{align}
\end{subequations}
where $C$ is a generic constant independent of the system size, and~\eqref{eq:Second zeta expand1}  follows since the matrices $\tilde{\vec{E}},\vec{A},\vec{T},\vec{E}$ have bounded spectral norms.
In~\eqref{eq:Second zeta expand3} we used the fact that $\overset{o}{q}_{i,i}$ terms are zero mean, where then, the positive correlations among the entries
yields the last inequality. Next, we use Nash-Poincar\'{e} inequality to find an upper-bound for the variance of $\phi=\tr(\vec{Q})$.
 In particular, observing that $|\frac{\partial \phi}{\partial {h}_{i,j}}|=|\frac{\partial \phi}{\partial {h}_{i,j}^*}|$, Nash-Poincar\'{e} inequality yields,
\begin{subequations}
	\label{eq:var beta}
\begin{align}
\var({\phi})&\leq 2  \sum_{i=1}^{F}\sum_{j=1}^{K} a_{j}^2 v_{i,j} \mathbb{E} \big| \frac{\partial { {\tr}( {\bf Q}  )}}{\partial {h}_{i,j}}\big|^2\\
&= 2  \sum_{i=1}^{F}\sum_{j=1}^{K} a_{j}^2 v_{i,j} \mathbb{E} \bigg| \sum_{p=1}^{F}\frac{\partial q_{p,p}}{\partial {h}_{i,j}}\bigg|^2\\
&= 2  \sum_{i=1}^{F}\sum_{j=1}^{K} a_{j}^2 v_{i,j}\mathbb{E} \big|\frac{1}{\sigma^2}  [ { {\bf h}_j\herm \bf Q}{\bf Q}]_i\big|^2\\
&= \frac{2}{\sigma^4}\sum_{j=1}^{K}a_{j}^2 \mathbb{E}   { {\bf h}_j\herm \bf Q}{\bf Q}  {\bf V}_j {\bf Q} {\bf Q} {\bf h}_j\\
&\leq \mathbb{E}\frac{2 \|{\bf A}\|^2 \|{\bf Q}\|^4}{\sigma^4 d} \tr(\vec{HH}\herm)\\
&\leq C\frac{K}{d}. \label{eq:var beta x}
\end{align}
\end{subequations}
The inequality in~\eqref{eq:var beta x} follows since from the resolvent identity one can verify that $\|\vec{Q}\|\leq 1$. Also, it is can be verified that $\tr(\vec{HH}\herm)\leq K$.
Finally, putting the results from~\eqref{eq:Second zeta expand} and~\eqref{eq:var beta} together, we get $|\zeta(\sigma^2)|\leq \frac{C}{ d^2}$.

\subsection{Replacing the matrix $\vec{E}_F$ by the deterministic matrix $\vec{R}_F$ } 
\label{sec:approx of E}
In the last step of the proof, we need to show that the matrix $\bf E_F({\bf V}_F,\sigma^2)$ can be replaced by the deterministic matrix  ${\bf R}_F({\bf V}_F,\sigma^2)=\diag\{r_p({\bf V}_F,\sigma^2), \forall p\in[1,..., F]\}$ where the $r_p$ values are given as the unique positive solution of the following system of equations 
	\begin{align}
	\label{eq:dif of r and delta}
	&r_p= \frac{\sigma^2}{\sigma^2+\tilde{\delta}_p}, &&\tilde{\delta}_p={\tr}({\sbf \Lambda}_{p}^{(F)}{\bf A}_F\tilde{\bf R}_F ),&&&p=1,..., F,\\
	&\tilde{r}_j=\frac{\sigma^2}{\sigma^2+\delta_j}, &&\delta_{j}={a_{j}^2} {\tr}( {\bf V}_{j}^{(F)} {\bf R}_F  ),&&&j=1,..., K
	\end{align}
	where the matrices ${\sbf \Lambda}_{p}^{(F)}$, ${\bf A}_F$, and ${\bf V}_j^{(F)}$  are as defined in~\eqref{eq:QHH expanded}. The superscript and subscript $F$ denote the dependency of the entities on the system size.   The matrix $\tilde{\bf R}_F$ is defined as $\tilde{\bf R}_F  ({\bf V}_F,\sigma^2)=\diag\{\tilde{r}_j({\bf V}_F,\sigma^2),\forall j\in[1,..., K]\}$.
Let us define $\vec{T}_F$ and ${\bf B}_F$  to be $ F\times  F$ and $ K\times  K$  diagonal deterministic matrices with uniformly bounded spectral norm. Then, we show that under the assumptions in the theorem, the following holds for every $\sigma^2\in \mathbb{R}^+$
	\begin{equation}
	\begin{aligned}
	\frac{1}{ F}\tr({\bf T}_F \bf{ E}_F({\bf V}_F,\sigma^2) )&=\frac{1}{ F}\tr({\bf T}_F \bf{ R}_F({\bf V}_F,\sigma^2) )+\mathcal{O}(\frac{1}{d^2}),\\
	\frac{1}{ F}\tr({\bf B}_F \tilde{\bf E}_F({\bf V}_F,\sigma^2) )&=\frac{1}{ F}\tr({\bf B}_F  \tilde{\bf R}_F({\bf V}_F,\sigma^2) )+\mathcal{O}(\frac{1}{ d^2}).
	\end{aligned}
	\end{equation}
	In the following, the dimension-superscript and subscript $F$ are omitted.
In proving the results we need to develop a well-quantified bound on the difference of the trace term  $\frac{1}{ F}\tr ({\bf T} ({\bf E} -{\bf R} ))$ as the dimensions grow large. First, by a mere development, we have
\begin{equation}
\label{eq:trDER bound}
\begin{aligned}
\frac{1}{ F}\big|\tr ({\bf T}({\bf E}-{\bf R}))\big|&=
\frac{1}{ F}\big|\tr ({\bf T}{\bf E}({\bf R}^{-1}-{\bf E}^{-1}){\bf R})\big|\\
&=\frac{1}{\sigma^2 F}\big|\sum_{p=1}^{ F} t_p e_p r_p(\tilde{\delta}_p-\tilde{\alpha}_p)\big|\\
&\leq \frac{\|{\bf T}\|}{\sigma^2 F}\sum_{p=1}^{ F} \big|\tilde{\delta}_p-\tilde{\alpha}_p\big|
\end{aligned}
\end{equation}
where the last equality follows from the upper bounds $\|{\bf R}\|\leq 1$ and $\|{\bf E}\|\leq 1$, implied by the definitions of $\bf R$ and $\bf E$ in~\eqref{eq:dif of r and delta} and~\eqref{eq:TrDQ plus var}, respectively. Under the same arguments, the term  $|\tilde{\delta}_p-\tilde{\alpha}_p|$ can be bounded as
\vspace{-0.5cm}
\begin{equation}
\label{eq:proof 1}
\begin{aligned}
|\tilde{\delta}_p-\tilde{\alpha}_p|&=\frac{1}{\sigma^2}\bigg|\sum_{j=1}^{ K}a_{j}^2  v_{p,j} (\alpha_j-\delta_j) \tilde{e}_j \tilde{r}_j\bigg|\\
&{\leq} \frac{1}{\sigma^2}\sum_{j=1}^{ K}a_{j}^2  v_{p,j} |\alpha_j-\delta_j| .
\end{aligned}
\end{equation} 
On the other hand, expanding $|\alpha_j-\delta_j|$, we have
\begin{equation}
\label{eq:proof 2}
\begin{aligned}
|\alpha_j-\delta_j|&={a_{j}^2}\big| \mathbb{E}\tr({\bf V}_{j}{\bf Q})-\sum_{p=1}^{ F} \frac{v_{p,j}}{1+\tilde{\delta}_p/\sigma^2}   \big|\\
&\overset{(a)}{=}{a_{j}^2} \big|\sum_{p=1}^{ F} v_{p,j}\big( \frac{1}{1+\tilde{\alpha}_p/\sigma^2} - \frac{1}{1+\tilde{\delta}_p/\sigma^2}   \big)\big|+\mathcal{O}(\frac{1}{d^2})\\
&\overset{(b)}{\leq} \frac{a_{j}^2}{\sigma^2} \sum_{p=1}^{ F} v_{p,j}|\tilde{\delta}_p-\tilde{\alpha}_p|+\mathcal{O}(\frac{1}{d^2})
\end{aligned}
\end{equation}
where equality (a) follows from~\eqref{eq:TrDQ plus var}, and (b) is given due to the bounds on the spectral norm of $\bf R$ and $\bf E$. 
These inequalities together yield,
\begin{equation}
\label{eq:temp proof tr 1}
\begin{aligned}
\frac{1}{ K}\sum_{j=1}^{ K}|\alpha_j-\delta_j|&\overset{\eqref{eq:proof 2}}{\leq}   \frac{\|{\bf A}\|}{\sigma^2 d}   \sum_{p=1}^{ F}
|\tilde{\delta}_p-\tilde{\alpha}_p|+\mathcal{O}(\frac{1}{d^2})\\
&\overset{\eqref{eq:proof 1}}{\leq} \frac{\|{\bf A}\|^2}{\sigma^4 d}   \sum_{p=1}^{ F} \sum_{j=1}^{ K}  v_{p,j} |\alpha_j-\delta_j|+\mathcal{O}(\frac{1}{d^2})\\
&=   \frac{P\|{\bf A}\|^2}{\sigma^4 d}    \sum_{j=1}^{ K}   |\alpha_j-\delta_j|+\mathcal{O}(\frac{1}{d^2})
\end{aligned}
\end{equation}
where the last equality follows since $ \sum_{p=1}^{ F} v_{p,j}\leq P$ where $P=\max\{P_k,\forall k\in\mathcal{K}\}$. Let us do a change of variable as $z=1/\sigma^2$ and express all the related functions in terms of $z$.
 Then, from the inequality in~\eqref{eq:temp proof tr 1}, we get
\begin{equation}
\label{eq:proof bound on delta minus alpha}
\begin{aligned}
(1-\frac{K}{d}{z^2 P\|{\bf A}\|^2})\frac{1}{ K}\sum_{j=1}^{ K} |{\alpha}_j-{\delta}_j|= \mathcal{O}(\frac{1}{d^2}).\\
\end{aligned}
\end{equation}
This inequality ensures that there exist a $z_0$ value such that $(1-\frac{K}{d}{z^2P\|{\bf A}\|^2})>0$. This further implies that $\frac{1}{ K}\sum_{j=1}^{ K} |{\alpha}_j-{\delta}_j|= \mathcal{O}(\frac{1}{d^2})$ holds for $z\leq z_0$. Thus, once we prove that  the above inequality holds for all $z\in\mathbb{R}^+$, the convergence of the trace term to the deterministic equivalent can be claimed. To do so, we consider $\alpha_{j}={a_{j}^2}\mathbb{E} {\tr}( {\bf V}_{j} {\bf Q}  )$ as a function in $z$ with extended domain from $z\in\mathbb{R^+}$ to $z\in\mathbb{C}\backslash\mathbb{R}^-$. It can be shown that the following integral representation for $\alpha_j(z)$ holds (see \cite[Appendix A]{DupuyLoubaton2011}),
\begin{equation}
\alpha_j(z)=\int_{0}^{+\infty}\frac{\mu_j(d\lambda)}{1+z\lambda}
\end{equation}
where $\mu_j$ is a uniquely defined positive measure on $\rr{+}$ such that $\mu_j(\rr{+})={a^2_{j}}\tr {\bf{V}}_j$. 
Based on the properties of Stieltjes transform~\cite[Theorem 3.2]{RMT}, $\alpha_j(z)$ can be upper-bounded as
\begin{equation}
\begin{aligned}
\alpha_j(z)\leq {a^2_{j}}\tr {\bf{V}}_j \frac{1}{|z|}\frac{1}{\rm{dist}(-\frac{1}{z},\rr{+})}\\
\leq \|\bf{A}\|\frac{1}{|z|}\frac{1}{\rm{dist}(-\frac{1}{z},\rr{+})}.
\end{aligned}
\end{equation}
Similarly,  a bound on $\delta_k(z)$ can be developed using integral representation.
This analysis shows that the functions $\delta_k(z)$ and $\alpha_k(z)$ belongs to the class of Stieltjes transforms of finite
positive measures carried by $\rr{+}$. Thus, $|\alpha_{k}(z)-\delta_{k}(z)|$ belongs to a family of analytic functions, which are bounded on any compact subset of $\cc\backslash\rr{-}$. As the result, the Vitali's convergence theorem~\cite{RMT} ensures that $\frac{1}{ K}\sum_{j=1}^{ K} |{\alpha}_j-{\delta}_j|$ goes to zero for any $z\in\mathbb{R}^+$ as $F \rightarrow \infty$.  What remains is to show that  the derived convergence rate $\mathcal{O}(\frac{1}{d^2})$ holds for all $z\in\mathbb{R}^+$, which completes the proof. The results follows from straightforward calculus where the reader is invited to refer to~\cite[Appendix C]{DupuyLoubaton2011} for the details about the derivation steps.

\section{Proof of Proposition~\ref{prop:KKT properties}}
\label{sec: Proof of kkt propertis}
The first property in the proposition follows directly since the mutual information is a strictly increasing function of the UEs' powers. One can also drive the same conclusion based on~\eqref{eq:KKT3} and~\eqref{eq:KKT1}. The former implies that $\delta_k\geq \frac{1}{\sigma^2F}a_k^2 \tilde{r}_k r_f >0$, and thus, the latter gives $\sum_{i \in \mathcal{F}} {v}_{i,k}=P_k$ for all UEs. The second property in the proposition can be justified by noticing that~\eqref{eq:KKT3} is satisfied for a UE $k$ only if we set $\delta_k^*= \frac{1}{\sigma^2F}a_k^2 \tilde{r}_k \max \{r_f\}$. Hence, from~\eqref{eq:KKT2}, we observe that UE $k$ transmits only on the sub-channels with largest $r_f$ values. Since, all other UEs also have the same preference, the condition in~\eqref{eq:KKT2} and~\eqref{eq:KKT3} are satisfied only if the UEs assign $v_{f,k}$ values such that $r_f=r,\forall f$. The value $r$ can be shown that is unique, i.e., $r=r^*$ for any solution to~\eqref{eq:Opt Asymp Rewriten}. Assume the matrix ${\vec V}$ to be a solution to~\eqref{eq:Opt Asymp Rewriten} that results in $r_f$ values to be equalized across sub-channels, i.e., $r_f=r,\forall f$. Plugging the ${\vec V}$ entries into~\eqref{eq:rf and rk}, we get
\begin{equation}
\tilde{r}_k=\frac{1}{1+   \frac{a_k^2}{\sigma^2}  \sum_{f\in \mathcal{F}} v_{f,k}^2  {r}},\forall k\in \mathcal{K}
\end{equation}
where, from the first property in the proposition, we have $\sum_{i \in \mathcal{F}} {v}_{i,k}=P_k,\forall k\in \mathcal{K}$, and thus,
\begin{equation}
\tilde{r}_k=\frac{\sigma^2}{\sigma^2+  {P_k} a_k^2    {r}},\forall k\in \mathcal{K}.
\end{equation}
Given the values of $\tilde{r}_k,\,\forall k\in\mathcal{K}$ as above, the value of $r$  can be evaluated from the following system of equations
\begin{equation}
\label{eq: r equation}
r=\big(1+\sum_{k\in\mathcal{K}}\frac{a_k^2  v_{f,k}}{\sigma^2+  {P_k}a_k^2   {r}}   \big)^{-1},\quad\forall f\in\mathcal{F}
\end{equation}
that holds for all $f\in\mathcal{F}$. Inverting the sides of~\eqref{eq: r equation} and summing over all $f\in\mathcal{F}$ results in
\begin{equation}
\frac{F}{r}=F+\sum_{k\in\mathcal{K}}\frac{a_k^2  }{\sigma^2+  {P_k}a_k^2   {r} } \sum_{f\in\mathcal{F}}{ v_{f,k}}
\end{equation}
where from the first property we have $\sum_{i \in \mathcal{F}} {v}_{i,k}=P_k,\forall k\in \mathcal{K}$, which gives $r=r^*$ as in~\eqref{eq:opt r}.


%
%
%
%
%
%



\bibliographystyle{IEEEtran}

\bibliography{Jourbib,Jourbibrew}

\begin{thebibliography}{10}
\providecommand{\url}[1]{#1}
\csname url@samestyle\endcsname
\providecommand{\newblock}{\relax}
\providecommand{\bibinfo}[2]{#2}
\providecommand{\BIBentrySTDinterwordspacing}{\spaceskip=0pt\relax}
\providecommand{\BIBentryALTinterwordstretchfactor}{4}
\providecommand{\BIBentryALTinterwordspacing}{\spaceskip=\fontdimen2\font plus
\BIBentryALTinterwordstretchfactor\fontdimen3\font minus
  \fontdimen4\font\relax}
\providecommand{\BIBforeignlanguage}[2]{{%
\expandafter\ifx\csname l@#1\endcsname\relax
\typeout{** WARNING: IEEEtran.bst: No hyphenation pattern has been}%
\typeout{** loaded for the language `#1'. Using the pattern for}%
\typeout{** the default language instead.}%
\else
\language=\csname l@#1\endcsname
\fi
#2}}
\providecommand{\BIBdecl}{\relax}
\BIBdecl

\bibitem{5GEvol19NOKIA}
A.~{Ghosh}, A.~{Maeder}, M.~{Baker}, and D.~{Chandramouli}, ``{5G} evolution: A
  view on {5G} cellular technology beyond {3GPP} release 15,'' \emph{IEEE
  Access}, vol.~7, pp. 127\,639--127\,651, 2019,
  doi:10.1109/ACCESS.2019.2939938.

\bibitem{hoshyar2008novel}
R.~Hoshyar, F.~P. Wathan, and R.~Tafazolli, ``Novel low-density signature for
  synchronous {CDMA} systems over {AWGN} channel,'' \emph{IEEE Transactions on
  Signal Processing}, vol.~56, no.~4, pp. 1616--1626, 2008.

\bibitem{SCMA14}
M.~{Taherzadeh}, H.~{Nikopour}, A.~{Bayesteh}, and H.~{Baligh}, ``{SCMA}
  codebook design,'' in \emph{Vehicular Technology Conference (VTC2014-Fall)},
  Sep. 2014, pp. 1--5.

\bibitem{MUSIot16}
Z.~{Yuan}, G.~{Yu}, W.~{Li}, Y.~{Yuan}, X.~{Wang}, and J.~{Xu}, ``Multi-user
  shared access for internet of things,'' in \emph{IEEE 83rd Vehicular
  Technology Conference (VTC Spring)}, May 2016, pp. 1--5.

\bibitem{NOMAsurveyDing17}
Z.~{Ding}, X.~{Lei}, G.~K. {Karagiannidis}, R.~{Schober}, J.~{Yuan}, and V.~K.
  {Bhargava}, ``A survey on non-orthogonal multiple access for {5G} networks:
  Research challenges and future trends,'' \emph{IEEE Journal on Selected Areas
  in Communications}, vol.~35, no.~10, pp. 2181--2195, Oct 2017.

\bibitem{NOMAreviewDaiWang2015}
L.~Dai, B.~Wang, Y.~Yuan, S.~Han, C.-l. I, and Z.~Wang, ``Non-orthogonal
  multiple access for 5{G}: Solutions, challenges, opportunities, and future
  research trends,'' \emph{IEEE Communications Magazine}, vol.~53, pp. 74--81,
  09 2015.

\bibitem{MCNomaUp14}
M.~{Al-Imari}, P.~{Xiao}, M.~A. {Imran}, and R.~{Tafazolli}, ``Uplink
  non-orthogonal multiple access for 5{G} wireless networks,'' in \emph{11th
  International Symposium on Wireless Communications Systems (ISWCS)}, Aug
  2014, pp. 781--785.

\bibitem{LDSOFDM-Hoshyar10}
R.~{Hoshyar}, R.~{Razavi}, and M.~{Al-Imari}, ``{LDS-OFDM} an efficient
  multiple access technique,'' in \emph{Vehicular Technology Conference}, May
  2010, pp. 1--5.

\bibitem{LDSOFDMRazavi12}
R.~{Razavi}, M.~{AL-Imari}, M.~A. {Imran}, R.~{Hoshyar}, and D.~{Chen}, ``On
  receiver design for uplink low density signature {OFDM} {(LDS-OFDM)},''
  \emph{IEEE Transactions on Communications}, vol.~60, no.~11, pp. 3499--3508,
  November 2012.

\bibitem{PeakToAvgImran10}
M.~{Al-Imari} and R.~{Hoshyar}, ``Reducing the peak to average power ratio of
  {LDS-OFDM} signals,'' in \emph{International Symposium on Wireless
  Communication Systems}, Sep. 2010, pp. 922--926.

\bibitem{ITLDSsurveyFerrant18}
M.~T.~P. {Le}, G.~C. {Ferrante}, G.~{Caso}, L.~{De Nardis}, and M.~{Di
  Benedetto}, ``On information-theoretic limits of code-domain {NOMA} for
  {5G},'' \emph{IET Communications}, vol.~12, no.~15, pp. 1864--1871, 2018.

\bibitem{LDSReviewImran12}
M.~Al-Imari, M.~Imran, and R.~Tafazolli, ``Low density spreading multiple
  access,'' \emph{Information Technology Software Engineering}, vol.~2, 09
  2012.

\bibitem{shitz2017}
O.~{Shental}, B.~M. {Zaidel}, and S.~S. {Shitz}, ``Low-density code-domain
  {NOMA}: Better be regular,'' in \emph{IEEE International Symposium on
  Information Theory (ISIT)}, June 2017, pp. 2628--2632.

\bibitem{zaidelBenj18}
B.~M. {Zaidel}, O.~{Shental}, and S.~S. {Shitz}, ``Sparse {NOMA}: A closed-form
  characterization,'' in \emph{IEEE International Symposium on Information
  Theory (ISIT)}, June 2018, pp. 1106--1110.

\bibitem{gerrante15SE}
G.~C. {Ferrante} and M.~D. {Benedetto}, ``Spectral efficiency of random
  time-hopping {CDMA},'' \emph{IEEE Transactions on Information Theory},
  vol.~61, no.~12, pp. 6643--6662, Dec 2015.

\bibitem{4036396}
M.~{Yoshida} and T.~{Tanaka}, ``Analysis of sparsely-spread {CDMA} via
  statistical mechanics,'' in \emph{IEEE International Symposium on Information
  Theory}, July 2006, pp. 2378--2382.

\bibitem{1633802}
A.~{Montanari} and D.~{Tse}, ``Analysis of belief propagation for non-linear
  problems: The example of {CDMA} (or: How to prove tanaka's formula),'' in
  \emph{IEEE Information Theory Workshop - ITW '06 Punta del Este}, March 2006,
  pp. 160--164.

\bibitem{7079688}
R.~R. {Müller}, ``Random matrices, free probability and the replica method,''
  in \emph{12th European Signal Processing Conference}, Sep. 2004, pp.
  189--196.

\bibitem{verdu1999spectral}
S.~Verd{\'u} and S.~Shamai, ``Spectral efficiency of {CDMA} with random
  spreading,'' \emph{IEEE Transactions on Information theory}, vol.~45, no.~2,
  pp. 622--640, 1999.

\bibitem{RMT}
R.~Couillet and M.~Debbah, \emph{Random Matrix Methods for Wireless
  Communications}.\hskip 1em plus 0.5em minus 0.4em\relax Cambridge University
  Press, 2011.

\bibitem{wood2012}
P.~M. Wood, ``Universality and the circular law for sparse random matrices,''
  \emph{Ann. Appl. Probab.}, vol.~22, no.~3, pp. 1266--1300, 06 2012.

\bibitem{FerranteFadingLDS18}
M.~T.~P. {Le}, G.~C. {Ferrante}, T.~Q.~S. {Quek}, and M.~{Di Benedetto},
  ``Fundamental limits of low-density spreading {NOMA} with fading,''
  \emph{IEEE Transactions on Wireless Communications}, vol.~17, no.~7, pp.
  4648--4659, July 2018.

\bibitem{RazaviHoshyarInfoLDS11}
R.~{Razavi}, R.~{Hoshyar}, M.~A. {Imran}, and Y.~{Wang}, ``Information
  theoretic analysis of lds scheme,'' \emph{IEEE Communications Letters},
  vol.~15, no.~8, pp. 798--800, 2011.

\bibitem{MullerTulino04a}
R.~{Müller} and A.~M. {Tulino}, ``Minimum bit error probability of large
  randomly spread {MC-CDMA} systems in multipath rayleigh fading,'' in
  \emph{Eighth IEEE International Symposium on Spread Spectrum Techniques and
  Applications - Programme and Book of Abstracts (IEEE Cat. No.04TH8738)}, Aug
  2004, pp. 560--564, doi:10.1109/ISSSTA.2004.1371762.

\bibitem{TseZeitLinearCdma99}
D.~N.~C. {Tse} and O.~{Zeitouni}, ``Performance of linear multiuser receivers
  in random environments,'' in \emph{IEEE Communications Theory Mini-Conference
  (Cat. No.99EX352)}, June 1999, pp. 163--167, doi=10.1109/CTMC.1999.790257.

\bibitem{TseHanlyLinearCDMA99}
D.~N.~C. {Tse} and S.~V. {Hanly}, ``Linear multiuser receivers: effective
  interference, effective bandwidth and user capacity,'' \emph{IEEE
  Transactions on Information Theory}, vol.~45, no.~2, pp. 641--657, March
  1999, doi=10.1109/18.749008.

\bibitem{GrantAlexOptimumCDMArandom98}
A.~J. {Grant} and P.~D. {Alexander}, ``Random sequence multisets for
  synchronous code-division multiple-access channels,'' \emph{IEEE Transactions
  on Information Theory}, vol.~44, no.~7, pp. 2832--2836, Nov 1998,
  10.1109/18.737515.

\bibitem{GrantAlexRandoopticDMA96}
------, ``Randomly selected spreading sequences for coded {CDMA},'' in
  \emph{Proceedings of ISSSTA'95 International Symposium on Spread Spectrum
  Techniques and Applications}, vol.~1, Sep. 1996, pp. 54--57 vol.1,
  doi=10.1109/ISSSTA.1996.563742.

\bibitem{shamaiverduCDMARand01}
S.~{Shamai} and S.~{Verdu}, ``The impact of frequency-flat fading on the
  spectral efficiency of {CDMA},'' \emph{IEEE Transactions on Information
  Theory}, vol.~47, no.~4, pp. 1302--1327, May 2001, doi=10.1109/18.923717.

\bibitem{HosISIT020}
H.~{Asgharimoghaddam} and A.~{Tölli}, ``Resource allocation in low density
  spreading uplink {NOMA} via asymptotic analysis,'' in \emph{IEEE
  International Symposium on Information Theory (ISIT)}, Los Angeles,
  California, USA, June 2020.

\bibitem{Kaiserbook}
K.~Fazel and S.~Kaiser, \emph{Multi‐Carrier and Spread Spectrum Systems: From
  {OFDM} and {MC‐CDMA} to {LTE} and {WiMAX}, Second Edition}.\hskip 1em plus
  0.5em minus 0.4em\relax John Wiley \& Sons, 2008, vol.~2.

\bibitem{BjornsonLuca2016}
E.~Björnson, L.~Sanguinetti, and M.~Debbah, ``Massive {MIMO} with imperfect
  channel covariance information,'' in \emph{Asilomar Conference on Signals,
  Systems and Computers}, Nov 2016, pp. 974--978.

\bibitem{Rupf94}
M.~Rupf and J.~L. Massey, ``Optimum sequence multisets for synchronous
  code-division multiple-access channels,'' \emph{IEEE Transactions on
  Information Theory}, vol.~40, no.~4, pp. 1261--1266, July 1994.

\bibitem{Viswanath99}
P.~Viswanath and V.~Anantharam, ``Optimal sequences and sum capacity of
  synchronous {CDMA} systems,'' \emph{IEEE Transactions on Information Theory},
  vol.~45, no.~6, pp. 1984--1991, Sept 1999.

\bibitem{Marshallbook79}
A.~W.~Marshall and I.~Olkin, \emph{Inequalities: Theory of Majorization and its
  applications}, 01 1979, vol. 143.

\bibitem{ReinhardtCDMAFading95}
M.~{Reinhardt} and J.~{Lindner}, ``Transformation of a {R}ayleigh fading
  channel into a set of parallel {AWGN} channels and its advantage for coded
  transmission,'' \emph{Electronics Letters}, vol.~31, no.~25, pp. 2154--2155,
  Dec 1995.

\bibitem{AWGNCApTeltar99}
\BIBentryALTinterwordspacing
E.~Telatar, ``Capacity of multi-antenna {Gaussian} channels,'' \emph{European
  Transactions on Telecommunications}, vol.~10, no.~6, pp. 585--595, 1999.
  [Online]. Available:
  \url{https://onlinelibrary.wiley.com/doi/abs/10.1002/ett.4460100604}
\BIBentrySTDinterwordspacing

\bibitem{Hachem07}
W.~Hachem, O.~Khorunzhiy, P.~Loubaton, J.~Najim, and L.~Pastur, ``A new
  approach for capacity analysis of large dimensional multi-antenna channels,''
  \emph{IEEE Transactions on Information Theory}, vol.~54, 01 2007.

\bibitem{Mller2013ApplicationsOL}
R.~R. M{\"u}ller, G.~Alfano, B.~M. Zaidel, and R.~de~Miguel, ``Applications of
  large random matrices in communications engineering,'' \emph{ArXiv}, vol.
  abs/1310.5479, 2013.

\bibitem{ReplicamethodMullerBook}
R.~R. Müller, ``The replica method in multiuser communications,'' in
  \emph{Random Matrix Theory and Its Applications}, July 2009, pp. 139--165,
  doi = 10.1142/97898142731210005.

\bibitem{hachemCLT2008}
W.~Hachem, P.~Loubaton, and J.~Najim, ``A {CLT} for information-theoretic
  statistics of {Gram} random matrices with a given variance profile,''
  \emph{Ann. Appl. Probab.}, vol.~18, no.~6, pp. 2071--2130, 12 2008,
  doi:10.1214/08-AAP515.

\bibitem{Boyd-Vandenberghe-04}
S.~Boyd and L.~Vandenberghe, \emph{Convex Optimization}.\hskip 1em plus 0.5em
  minus 0.4em\relax Cambridge University Press, 2004.

\bibitem{AntonChainRule95}
H.~Anton, \emph{Calculus with Analytic Geometry}.\hskip 1em plus 0.5em minus
  0.4em\relax John Wiley and Sons Inc., 1995, vol.~10.

\bibitem{CDMAcodingTradeoff02}
V.~V. {Veeravalli} and A.~{Mantravadi}, ``The coding-spreading tradeoff in
  {CDMA} systems,'' \emph{IEEE Journal on Selected Areas in Communications},
  vol.~20, no.~2, pp. 396--408, Feb 2002, doi=10.1109/49.983362.

\bibitem{BrucksteinSparseSol}
A.~M. Bruckstein, D.~L. Donoho, and M.~Elad, ``From sparse solutions of systems
  of equations to sparse modeling of signals and images,'' \emph{SIAM Rev.},
  vol.~51, no.~1, pp. 34--81, Feb. 2009.

\bibitem{KORF1998181}
R.~E. Korf, ``A complete anytime algorithm for number partitioning,''
  \emph{Artificial Intelligence}, vol. 106, no.~2, pp. 181 -- 203, 1998.

\bibitem{EasyNPpoblm}
B.~Hayes, ``Computing science: The easiest hard problem,'' \emph{American
  Scientist}, vol.~90, no.~2, pp. 113--117, 2002.

\bibitem{Graham1969BoundsOM}
R.~L. Graham, ``Bounds on multiprocessing timing anomalies,'' \emph{SIAM
  Journal of Applied Mathematics}, vol.~17, pp. 416--429, 1969.

\bibitem{DupuyLoubaton2011}
F.~Dupuy and P.~Loubaton, ``On the capacity achieving covariance matrix for
  frequency selective {MIMO} channels using the asymptotic approach,''
  \emph{IEEE Transactions on Information Theory}, vol.~57, no.~9, pp.
  5737--5753, Sept 2011.

\end{thebibliography}
%
%
%

\end{document}